\documentclass[12pt]{article}

\usepackage[T1]{fontenc}
\usepackage[utf8]{inputenc}
\usepackage{xcolor}
\usepackage{verbatim}
\usepackage{amsmath}
\usepackage{amsthm}
\usepackage{amssymb}
\usepackage{graphicx}
\usepackage{geometry}
\usepackage{setspace}
\usepackage{esint}
\usepackage[authoryear]{natbib}
\usepackage[bookmarks=true,bookmarksnumbered=false,bookmarksopen=false,
 breaklinks=true,pdfborder={0 0 1},backref=section,colorlinks=true]
 {hyperref}
\hypersetup{
 linkcolor=blue,urlcolor=blue,citecolor=blue,setpagesize=false,linktocpage}

\makeatletter

\theoremstyle{plain}
\newtheorem{assumption}{\protect\assumptionname}
\theoremstyle{plain}
\newtheorem{thm}{\protect\theoremname}
\theoremstyle{plain}
\newtheorem{prop}{\protect\propositionname}
\theoremstyle{plain}
\newtheorem{lem}{\protect\lemmaname}

\usepackage[margin=28pt,font=small,labelsep=endash]{caption}
\usepackage{amsfonts}
\usepackage{authblk}
\providecommand{\U}[1]{\protect\rule{.1in}{.1in}}
\ifdefined\showcaptionsetup
 
\fi

\providecommand{\assumptionname}{Assumption}
\providecommand{\lemmaname}{Lemma}
\providecommand{\propositionname}{Proposition}

\providecommand{\theoremname}{Theorem}

\providecommand{\Vtz}{\boldsymbol{z}}
\usepackage{etoolbox}

\pretocmd\@maketitle
  {\def\@makefnmark{\hbox{\@textsuperscript{\normalfont\@thefnmark}}}}
  {}{\FAIL}

\ifdefined\showcaptionsetup
 \PassOptionsToPackage{caption=false}{subfig}
\fi
\usepackage{subfig}
\makeatother

\providecommand{\assumptionname}{Assumption}
\providecommand{\lemmaname}{Lemma}
\providecommand{\propositionname}{Proposition}
\providecommand{\theoremname}{Theorem}

\begin{document}
\title{{\Large Innovation, Spillovers and Economic Geography}}
\author{José M. Gaspar\thanks{CEF.UP, Faculty of Economics, University of Porto. Email: jgaspar@fep.up.pt.}
and Minoru Osawa\thanks{Institute of Economic Research, Kyoto University. Email: osawa.minoru.4z@kyoto-u.ac.jp.}}
\date{\vspace{-5ex}
 }
\maketitle
\begin{abstract}
We develop a Schumpeterian quality-ladder spatial model in which innovation
arrivals depend on regional knowledge spillovers. A parsimonious reduced-form
diffusion mechanism induces the convergence of regions' average distance
to the global frontier quality. As a result, regional differences
in knowledge levels stem residually from asymmetries in the spatial
distribution of researchers and firms. We analytically characterize
the processes of innovation and knowledge diffusion. We then explore
how the weight of intra- relative to inter-regional knowledge spillovers
interacts with freer trade to shape the spatial distribution of economic
activities. If intra-regional spillovers are relatively stronger,
a higher economic integration leads to progressive agglomeration.
If inter-regional spillovers dominate, researchers and firms may re-disperse
after an initial phase of agglomeration as integration increases.
This happens because firms and researchers have incentives to relocate
to the smaller region, where they can leverage the concentrated knowledge
base of the larger region while avoiding congestion in innovation.
The smoothness of the dispersion process depends on the particular
weight of intra-regional spillovers. If inter-regional spillovers
become stronger as trade becomes freer, then the latter induces a
monotone dispersion process. When integration is high enough, stable
long-run equilibria always maximize the growth rate of the global
frontier quality and the average distance to the frontier, irrespective
of whether spillovers are mainly local or global. 
\end{abstract}
\bigskip{}

\noindent\textbf{Keywords: }Innovation; regional knowledge spillovers;
economic geography; dispersion;

\noindent\textbf{JEL codes: }O30, R10, R12, R23.

\section{Introduction}

Economic geography emphasizes the role of endogenous forces in shaping
lasting and sizable economic agglomerations in the modern economy.
In its aim to explain the spatial distribution of economic activities,
there has been a narrow focus on pecuniary externalities through trade
linkages.\footnote{Reviews on the literature of geographical economics, or ``new economic
geography'' models, are provided by e.g. \citet{Fujita2013}
and \citet{gaspar2018prospective}.} However, in order to understand the processes of agglomeration, dispersion
and (de)-industrialization, it is crucial to develop theories that
explore the interaction among multiple spatial linkages \citep{gaspar2018prospective}.
We aim to fill this gap by explaining how the weight/intensity of
intra-regional spillovers impacts knowledge creation and affects the
spatial distribution of agents. Moreover, we study how this weight
interplays with economic integration to understand the evolution of
the space economy as trade barriers decrease.

We develop a two-region spatial quality-ladder model that is based
on microfoundations borrowed from the literature of endogenous growth
theory \citep{Aghion-Howitt-Book1998}, and augment it with a regional
spillover function such that regional innovation rates depend on the
geography of knowledge spillovers \citep{audretsch2004knowledge},
i.e.\ on how the spatial distribution of mobile agents (researchers)
and firms in the economy affects regional knowledge spillovers. Therefore,
it is assumed that knowledge transfers imperfectly between regions,
depending on the spatial nature of knowledge spillovers. We analytically
characterize the processes of innovation and knowledge diffusion.

We introduce a parsimonious reduced-form diffusion mechanism in the
dynamics of average qualities (relative to the global frontier), which
induces the convergence of regions' average distance to the global
frontier quality. As a result, regional differences in aggregate knowledge
levels stem residually from asymmetries in the spatial distribution
of researchers and firms. This allows us to focus on a setting where
regions are structurally symmetric. In turn, the spatial distribution
of researchers is determined endogenously by the interplay between
trade linkages and knowledge linkages through regional spillovers.
This confers great analytical tractability and allows us to focus
on spatial outcomes, based on first principles, as a result of pecuniary
factors and the \emph{economic geography} of knowledge spillovers
\citep{BS2022}.

A key issue in the geography of innovation is whether knowledge spillovers
are mainly local or extend across regions. When innovation depends
heavily on tacit know-how and face-to-face interaction, local spillovers
dominate. The local nature of spillovers is supported by empirical
evidence, as summarized by \citet{Doring01052006}. However, with
the ongoing globalization of innovation, which promotes the transfer
of knowledge across different regions, agents will tend to leverage
externally generated knowledge more and more. Thus, theoretical and
empirical refinements of the economic geography of innovation have
become increasingly important. We discuss this issue further in Section
3.


We find that, when spillovers are mainly global, then an increase
in economic integration from a very low level initially fosters agglomeration
in a single region. However, above a certain threshold, more integration
leads to more symmetric spatial outcomes, because firms find it worthwhile
to relocate to the smaller region in order to benefit from the sizeable
pool of agents in the larger region, which increases the chance of
innovation and expected profits in the smaller region. Therefore,
when the relative weight of local spillovers is low (but not too low),\footnote{For exceedingly low values of the weight of local spillovers, the
symmetric dispersion equilibrium is the unique stable equilibrium
in the entire range of economic integration.} our model accounts for a (complete) re-dispersion of economic activities
after an initial phase of agglomeration. That is, we are able to uncover
a bell-shaped relation between economic integration and spatial development
\citep{Fujita2013}. In this case, knowledge spillovers constitute
a dispersion force that becomes relatively stronger as economic integration
brings about the withering of agglomeration forces due to increasing
returns to scale in manufacturing. But the process of (de)-industrialization
is not trivial; depending on the specific weight of local spillovers
relative to global spillovers, the relationship between economic integration
and the spatial distribution of activities occurs with very different
qualitative properties, not yet described in the literature.\footnote{To the best of our knowledge.}
In fact, we show that increases in the relative weight of local spillovers
are linked to more pronounced agglomerations in the industrialization
process, and to more sudden (discontinuous) jumps towards dispersed
outcomes, particularly for intermediate levels of economic integration.

By contrast, when spillovers are mainly local, they generate an additional
agglomeration force. Re-dispersion becomes altogether impossible because
within-region interaction is too important for innovation to make
any deviation to a smaller region worthwhile. In this case, a higher
trade integration increases the net agglomeration forces and promotes
agglomeration, in line with the findings of \citet{martin2001growth}.
Moreover, an increase in the relative weight of local spillovers benefits
agglomeration and leads to greater spatial inequality.

We find that the complete re-dispersion of economic activities requires
the weight of local spillovers to be low at the symmetric (dispersion)
equilibrium, irrespective of the particular functional form of regional
spillover function. This captures the idea of a congestive effect
in the production of knowledge, which makes innovation less likely
in larger regions and creates opportunities for breakthrough innovations
in smaller regions. This promotes the convergence of regions toward
a balanced distribution of economic activities, in line e.g. with
\citet{brezis1993}.

We also find that, if the weight of global spillovers increases as
trade barriers decrease \citep{keller2010international,bueraoberfield2020},
eventually dominating local spillovers, then the economy disperses
monotonically from the state of agglomeration as economic integration
increases. This change may be smooth or sudden, depending on the specific
weight of local knowledge spillovers.

We thus present a microfounded narrative that may explain why and
how the spatial nature of knowledge spillovers may mitigate or exacerbate
spatial asymmetries in increasingly integrated economies. As such,
our paper provides a new theoretical mechanism that could be incorporated
into larger-scale quantitative models, such as the ones discussed
in \citet{redding2017quantitative}, to better understand the determinants
of the spatial distribution of economic activity and derive counterfactuals.

The rest of the paper is organized as follows. Section 2 discusses
some related literature. Section 3 discusses the empirical literature
on the spatial nature of knowledge spillovers. Section 4 introduces
the spatial quality-ladder model. Section 5 characterizes the dynamics
of innovation. Section 6 deals with the existence and stability of
long-run equilibria. Section 7 studies the relationship between economic
integration and spatial outcomes. In Section 8, we provide some comparative
statics and in Section 9 we present results with a more general form
for the regional's innovation rate. In Section 10 we briefly discuss
the case where the weight of inter-regional spillovers increases as
trade becomes freer. Finally, Section 11 is left for discussion and
concluding remarks.

\section{Literature Review}

The role of knowledge linkages has become increasingly relevant in
the economic geography literature. According to \citet{Fujita-RSUE2007},
geography is an essential feature of knowledge creation and diffusion.
For instance, people residing in the same region interact more frequently
and thus contribute to develop the same, regional set of cultural
ideas. However, while each region tends to develop its unique culture,
the economy as a whole evolves according to the synergy that results
from the interaction across different regions. That is, according
to \citet{Duranton-Puga-AER2001}, knowledge creation and location
are inter-dependent. \citet{Berliant-Fujita-RSUE2012} developed a
model of spatial knowledge interactions and showed that higher cultural
diversity, albeit hindering communication, promotes the productivity
of knowledge creation. This corroborates the empirical findings of
\citet{Ottaviano-Peri-JEG2006,Ottaviano-Peri-DP2008}. \citet{Ottaviano-Prarolo-JRS2009}
show how improvements in the communication between different cultures
fosters the creation of multicultural cities in which cultural diversity
promotes productivity. This happens because better communication allows
different communities to interact and benefit from productive externalities
without risking losing their cultural identities.


\citet{Berliant-Fujita-SEJ2011} take a first step towards using a
micro-founded R\&D structure to infer about its effects on economic
growth. They find that long-run growth is positively related to the
effectiveness of interaction among workers as well as the effectiveness
in the transmission of public knowledge. 



In the endogenous growth models developed e.g. by \citet{Aghion-Howitt-ECTA1990,Aghion-Howitt-Book1998},
\citet{Young-JPE1998}, \citet{Howitt-JPE1999}, or more recently
\citet{Dinopoulous-Segerstrom-JDE2010}, innovation depends on factors
such as the amount of the firms' research effort, the common pool
of public knowledge available to all firms, and the individual firm's
quality level. Introducing geography and worker mobility in these
frameworks allows the innovation success to also depend on the magnitude
of regional interaction through the exchange of ideas between researchers
and producers alike among regions. This is the aim we pursue in the
present paper.

We introduce an innovation sector in an economic geography framework,
where consumers have quasi-linear preferences \citep{pfluger2004simple},
that is based on microfoundations borrowed from the literature of
endogenous growth theory \citep{Aghion-Howitt-Book1998}, but with
myopic agents (as in \citep[see e.g.][]{akcigit2022international,aghion2025theory}.
We augment it with a regional spillover function that depends on the
spatial distribution of regional knowledge levels. Under mild assumptions,
our modeling strategy is such that regional average qualities relative
to the frontier are equalized in the long-run, such that the structural
symmetry across regions is preserved. Any asymmetries in the model
thus stem from the spatial distribution of researchers and firms.
Indirect utility differentials, which govern the migration of researchers
between regions in the long-run, are determined solely by trade linkages
and by knowledge linkages.

Our modeling of the innovation sector provides a framework that is
similar in analytical tractability to the models found in the literature
of \emph{Quantitative Spatial Economics} (QSE) \citep{redding2017quantitative,BEHRENS2021103348,kleinman2023linear}.
However, the innovation process in our paper has the additional advantage
of being grounded on solid microfoundations and derived from first
principles. Moreover, in our setting, knowledge transfers imperfectly
between regions, in accordance with the empirical literature on geographical
knowledge spillovers and regional growth \citep{Doring01052006}.
Finally, our model is scale-neutral. As such, its spatial outcomes
are driven by the spatial nature of knowledge spillovers rather than
by implicit assumptions about returns to scale in the innovation sector.\footnote{For a comprehensive discussion on how these implicit assumptions regarding
returns to scale generate mistaken conclusions in geographical economics,
we refer the reader to \citet{BS2022}.} Thus, we argue that our paper may contribute to the literature by
providing a new theoretical mechanism that could be incorporated into
larger-scale quantitative models to better understand the determinants
of the spatial distribution of economic activity.

\section{Evidence on the geography of knowledge spillovers}

\label{subsec:spillover-mechanisms}

In this paper, the spatial nature of knowledge spillovers is embodied
by a parameter $b$, which captures the relative importance of intra-
versus inter-regional spillovers, measured by the aggregate knowledge
levels in the domestic region and in the foreign region. Formally,
it is the weight of local knowledge contribution on a region's innovation
rate. As we will see, it will depend directly on the spatial distribution
of researchers and firms across regions, thus capturing the idea that
interaction between researchers both within and across regions generates
spillovers that foster innovation. In terms of regional knowledge
spillovers, we can write 
\begin{equation}
b=\frac{\text{intra-regional spillovers}}{\text{total spillovers}}\quad\in(0,1),
\end{equation}
so that when $b>1/2$ local spillovers dominate, while when $b<1/2$
global spillovers dominate. The empirical literature provides indirect
evidence on the geography of spillovers by examining how knowledge
diffuses across space.

Early work emphasizes the localized character of spillovers, driven
by tacit knowledge, repeated interaction, and agglomeration economies.
Patent citation analyses \citep{jaffe1993geographic} and studies
of innovative clusters \citep{audretsch1996R&D} showed that knowledge
flows decay quickly with distance, especially for tacit know-how that
is hard to transfer remotely. \citet{storper2004buzz} similarly emphasize
the creative ``buzz'' in cities driven by in-person contacts and
informal exchanges. Such ``local buzz'' \citep{bathelt2004buzz}
is particularly relevant for incremental innovation, implying settings
where local spillovers dominate $(b>1/2)$.

Recent research demonstrates, however, that several mechanisms expand
the reach of spillovers. A first is codification, reinforced by digital
communication. Once knowledge is formalized in patents, publications,
or blueprints, it can circulate widely. \citet{keller2002geography}
finds that foreign R\&D embodied in trade accounts for a significant
share of productivity growth. \citet{jaffe2002patents} show that
patent citations cross regional and national borders; and \citet{forman2019internet}
demonstrate that basic internet adoption increased cross-location
patent citations. These results point to values of $b$ closer to
or below $0.5$ as communication technologies and codification allow
ideas to circulate globally.

A second mechanism combines global networks and researcher mobility.
\citet{branstetter2006is} shows that Japanese multinationals in the
U.S.\ generated two-way knowledge flows. \citet{almeida1999localization}
and \citet{agrawal2006mobility} find that migrating inventors transmit
knowledge across regions. Related work confirms that even international
mobility, such as emigrant inventors or visiting researchers, helps
bridge regional innovation gaps (e.g. returning diaspora entrepreneurs
spreading know-how). These ``people- and firm-based'' channels reduce
the weight of local spillovers implying $b<1/2$ in contexts with
high talent mobility.

A third mechanism is collaborative R\&D and complementarities across
regions. Large consortia and policy programs, such as the EU Framework
Programmes, foster cross-regional partnerships. Evidence from joint
patents \citep{hoekman2010research} shows that such collaborations
deliberately spread knowledge across regions. By constructing ``global
pipelines'' between clusters, these collaborations ensure that discoveries
in one region are shared with partners elsewhere \citep{branstetter2006is}.
Such network-building across space reduces the weight of local spillovers
and pushes $b$ below $0.5$.

Finally, the type of innovation also shapes the geography of knowledge
diffusion. Incremental improvements and innovations based on tacit
know-how tend to stay local, since they rely on context-specific expertise
and close interactions. By contrast, major breakthroughs and radical
innovations are usually codified (e.g. published or patented) and
thus diffuse more globally. For instance, studies of European regions
suggest that routine, tacit innovations generate mostly local knowledge
spillovers, whereas cutting-edge innovations create knowledge that
travels widely \citep{autant2013social}. This factor shifts $b$
above or below $0.5$ depending on the technological trajectory.

Overall, the evidence suggests a historical shift: while tacit and
incremental innovation sustains the dominance of local spillovers,
the forces of codification, mobility, networks, and collaboration
are pushing economies toward more globalized spillovers \citep{carlino2015agglomeration,miguelez2018relatedness}.
In the context of our two-region model, where the regions are structurally
similar, differing mainly in size or scale of activity, this justifies
a close examination of the case $b<1/2$.

\section{The model}

The following is a solvable Core-Periphery model with both horizontally
and vertically differentiated varieties of manufactured goods. The
economy is comprised of two regions indexed by $i\in\{1,2\}$, two
kinds of labour, two productive sectors and one R\&D sector. There
is a unit mass of (skilled) inter-regionally mobile agents (henceforth,
\emph{researchers}) and a mass $l\equiv\lambda>0$ of (unskilled)
immobile workers (\emph{workers} for short) which are assumed to be
evenly distributed across both regions, i.e., $l_{i}=\frac{\lambda}{2}$
$(i=1,2)$. Outsider firms engage in R\&D activities by employing
researchers in order to improve the quality of manufactured varieties.
The mass of researchers in region $i$ is denoted by $z_{i}\in\left[0,1\right]$
and their spatial distribution by $\Vtz=(z_{1},z_{2})$ with $z_{1}+z_{2}=1$.

The remainder of this section derives the short-run general equilibrium,
taking $\Vtz$ as fixed. We solve the model backwards, that is, we
start with the firm's production and pricing decisions for a given
variety, and then subsequently discuss the process of innovation and
the wages paid to researchers.

\subsection{Demand}

The utility function of a consumer located in region $i$ is given
by: 
\begin{equation}
u_{i}=\mu\ln\left(\frac{M_{i}}{\mu}\right)+B_{i},\ \ \mu>0\label{eq:utility}
\end{equation}
where $B_{i}$ is the numéraire good produced under perfect competition
and constant returns to scale. This good is produced one-for-one using
$L$ workers and its price is set to unity as is the wage paid to
workers. The quality-augmented CES composite $M_{i}$ is given by:
\begin{equation}
M_{i}=\left[\sum_{j}\int_{s\in\mathcal{S}_{j}}\left(\sum_{m=0}^{k}\delta^{m}d_{ji}(m,s)\right)^{\tfrac{\sigma-1}{\sigma}}\mathrm{d}s\right]^{\tfrac{\sigma}{\sigma-1}},\label{eq:ces utility}
\end{equation}
where $\mathcal{S}_{i}$ is the set of varieties in region $i$, $d_{ji}(m,s)$
is the demand for manufactures in region $i$ produced in region $j$,
for a given variety $s\in\mathcal{S}_{j}$ with quality $m\in\{0,1,...,k\}$,
and $\sigma>1$ is the elasticity of substitution between any two
varieties. The parameter $\delta>1$ indexes the step size of quality
improvements in region $i$ after a successful innovation and $k$
is the leading quality grade for any given variety $s$. Since $\delta^{m}$
is increasing in $m$, the utility in (\ref{eq:ces utility}) reflects
the fact that consumers have a preference for higher quality. Throughout
the paper, the summation over regional indices, e.g., $\sum_{j}(\textcolor{lightgray}{\blacksquare}_{j})$,
denotes $\sum_{j=1}^{2}(\textcolor{lightgray}{\blacksquare}_{j})$.

For each variety $s$, it is optimal for households to buy only the
good with the lowest quality-adjusted price, $p_{ji}(m,s)/\delta^{m}$
(\citealp{Dinopoulous-Segerstrom-JDE2010}). If any two goods of the
same variety have the same quality adjusted price, we assume that
consumers will only buy the highest quality good (see \citealp{Dinopoulous-Segerstrom-JDE2010,davis2012private}).

Since individual incomes depend on the distribution of labour activities,
we have $y_{i}=1$ as labour income for the workers, and $y_{i}=w_{i}$,
which is the compensation paid to the researchers. Therefore, the
regional income is given by: Agents maximize (\ref{eq:utility}) subject
to the following budget constraint 
\[
B_{i}+\sum_{j}\int_{s\in\mathcal{S}_{j}}p_{ji}(s)d_{ji}(s)\mathrm{d}s=y_{i}+\bar{B}_{i},
\]
where $d_{ji}(s)$ is the demand of the good with the lowest quality
adjusted price for variety $s$, $p_{ji}(s)$ is the corresponding
price and $\bar{B}_{i}>0$ is an endowment of the numéraire good to
be discussed later. This yields the following optimal individual demands:
\begin{equation}
d_{ji}(s)=\mu\frac{a_{ji}(s)p_{ji}(s)^{-\sigma}}{P_{i}^{1-\sigma}},\ \ \ \ B_{i}=y_{i}+\bar{B}_{i}-\mu,\ \ \ \ M_{i}=\mu P_{i}^{-1},\label{eq:optimal individual demand}
\end{equation}
where $a_{ji}(s)=\delta_{ji}^{m(s)(\sigma-1)}$ is an alternative
measure of the quality of a variety $s$ and $P_{i}$ is the quality
adjusted price index in region $i$ given by: 
\begin{equation}
P_{i}=\left[\sum_{j}\int_{s\in\mathcal{S}_{j}}a_{j}(s)p_{ji}(s)^{1-\sigma}\mathrm{d}s\right]^{\tfrac{1}{1-\sigma}}.\label{eq:quality price index}
\end{equation}

\noindent We assume that $\bar{B}>\mu$ in order to assure that both
types of goods are consumed. From (\ref{eq:utility}) and (\ref{eq:optimal individual demand}),
we obtain the indirect utility: 
\begin{equation}
v_{i}=y_{i}-\mu\ln P_{i}-\mu+\bar{B}_{i}.\label{eq:indirect utility}
\end{equation}

\subsection{Manufacturing firms}

The firm responsible for each quality improvement for a variety $s$
retains a monopoly right to produce that variety at the corresponding
quality level. Therefore, if the quality rungs $m=1,...,k$ have been
reached, the $m\textit{th}$ innovator is the sole source of the good
of variety $s$ with the quality level $\delta^{m}$ \citep{Barro-Sala-i-Martin-Book2004}.
As a result, each variety is produced by a single firm.

For each firm, there is a variable input requirement of $\beta$ workers.
A manufacturing firm in region $i$ thus faces the following cost:
\begin{equation}
C_{i}(q_{i}(s))=\beta q_{i}(s),\label{eq:total cost function}
\end{equation}
where $q_{i}(s)$ is total production by a firm in region $i$ that
produces the lowest quality adjusted price for variety $s$.

Interregion trade of manufactures is burdened by iceberg transportation
costs. Let $\tau>1$ denote the number of units that must be shipped
from region $i$ so that one unit is delivered at region $j\ne i$.
For convenience, we set $\tau_{ij}=\tau$ for $i\neq j$ and $\tau_{ij}=1$
otherwise. The quantity produced by a firm in region $i$ is then
given by: 
\[
q_{i}(s)=\sum_{j}\tau_{ij}d_{ij}(s)\left(\frac{\lambda}{2}+z_{j}\right).
\]

In the following, $s^{*}\in\mathcal{S}$ indexes a variety-$s$ firm
that produces at the leading quality grade $k$. Here, the precise
notation might be $s(k(s))$, since $k$ is $s$-dependent. To avoid
notational burden, our discussion below will omit $k$ and its dependence
on $s$, and further suppress $k$ by using $s^{*}$.

An $s^{*}$-producing firm faces competition from a potential producer
of the same variety with the next best grade. If innovations are drastic
(the step size of improvements $\delta$ is large enough), the quality
leader can charge the unconstrained monopoly price since the next
best producer cannot make a profit. If innovations are non-drastic,
the quality leader may initially engage in limit pricing and immediately
revert to the unconstrained monopoly price once it learns that the
next best producer has exited the market.\footnote{We assume that there are positive costs of re-entering the market.
This is why the quality leader can charge the unconstrained monopolist
price without worrying about competition from lower grades of the
product \citep{Howitt-JPE1999}. } In either case, the closest competitor (and all lower grades) is
priced out of business. In what follows, the short-run general equilibrium
will be comprised solely of the firms that are able to produce the
highest quality possible of each variety, i.e. the quality leaders.

The profit of a manufacturing firm producing the variety $s^{*}$
in region $i$ is given by: 
\begin{align}
\pi_{i}(s^{*})= & \sum_{j}p_{ij}(s^{*})d_{ij}(s^{*})\left(\frac{\lambda}{2}+z_{j}\right)-\beta q_{i}(s^{*})\nonumber \\
= & \sum_{j}\left(p_{ij}(s^{*})-\tau_{ij}\beta\right)d_{ij}(s^{*})\left(\frac{\lambda}{2}+z_{j}\right).\label{eq:profit of firm i}
\end{align}
Given (\ref{eq:profit of firm i}) and the optimal individual demand
in (\ref{eq:optimal individual demand}), the firm's profit maximizing
price is the usual mark-up over marginal cost: 
\begin{equation}
p_{ij}=\dfrac{\sigma}{\sigma-1}\tau_{ij}\beta,\label{eq:optimal price}
\end{equation}
which does not depend on the quality of the firm's variety $s$.

Under (\ref{eq:optimal price}), the regional quality adjusted price
index in (\ref{eq:quality price index}) becomes: 
\begin{equation}
P_{i}=\dfrac{\beta\sigma}{\sigma-1}\left(\sum_{j}\phi_{ij}A_{j}\right)^{\tfrac{1}{1-\sigma}},\label{eq:quality price index 3}
\end{equation}
where $\phi_{ij}\equiv\tau_{ij}^{1-\sigma}\in(0,1)$ is the \emph{freeness
of trade} and 
\begin{equation}
A_{j}=\int_{s^{*}\in\mathcal{S}_{j}}a_{j}(s^{*})\mathrm{d}s^{*},\label{eq:aggregateregionalknowledge}
\end{equation}
is the aggregate knowledge (quality) level in region $j$.

We can write $\pi_{i}(s^{*})$ as 
\begin{equation}
\pi_{i}(s^{*})=a_{i}(s^{*})\tilde{\pi}_{i},\label{eq:profitreduced}
\end{equation}
where 
\[
\tilde{\pi}_{i}=\sum_{j}\left(p_{ij}-\tau_{ij}\beta\right)\mu\frac{p_{ij}^{-\sigma}}{P_{i}^{1-\sigma}}\left(\frac{\lambda}{2}+z_{j}\right),
\]
is the profit of a firm in region $i$ with quality $m(s)=0$. Using
(\ref{eq:optimal individual demand}) and (\ref{eq:optimal price}),
we get 
\begin{equation}
\tilde{\pi}_{i}=\frac{\mu}{\sigma}\left(\frac{\frac{\lambda}{2}+z_{i}}{A_{i}+\phi A_{j}}+\phi\frac{\frac{\lambda}{2}+z_{j}}{\phi A_{i}+A_{j}}\right).\label{eq:qualityzeroprofit}
\end{equation}

\subsection{R\&D sector}



The modeling of the R\&D sector is based on the literature of Schumpeterian
growth theory (see e.g. \citealp{grossman1991trade,Aghion-Howitt-ECTA1990,Aghion-Howitt-Book1998,aghion2014we}),
except that we consider myopic agents (as in \citealt{akcigit2022international,aghion2025theory})
for simplicity and without loss of generality. There are $r$ (indeterminate)
symmetric outsider firms, funded by entrepreneurs, who engage in research
activities in order to improve the state-of-the-art quality of a chosen
variety $s$. There is free entry in the R\&D sector for each variety
$s$. Innovation in a variety occurs with Poisson intensity $\Phi_{i}(s)\equiv\Phi_{i}\geq0$.\footnote{We assume at most one success per variety at an instant; ties are
broken by an arbitrary rule.} In order to innovate, an outsider firm employs $\alpha/r$ researchers
and pays them the wage flow $w_{i}$. At the instant of success the
innovator upgrades the variety to quality $k(s)+1$, indexed by $s^{**}$
and immediately sells the patent/line to production firms (new entrants)
through a competitive patent market. The new quality leader displaces
the previous incumbent and earns subsequent profit flows. The patent
prize is equal to the instantaneous operating profit of the new leader,
$\pi_{i}(s^{**})$, determined by the profit function in (\ref{eq:profit of firm i})
evaluated at $s^{**}$.\footnote{With forward-looking agents, the \emph{present-value} reward of innovation
$V_{i}$ is the value of becoming and remaining the leader, given
by the standard Hamilton--Jacobi--Bellman asset pricing formula
$V_{i}=\frac{\pi_{i}(s^{**})}{\rho_{i}+\Phi_{i}}$, where $\rho_{i}>0$
is the discount rate in region $i$. If we choose time units such
that $\rho_{i}+\Phi_{i}=1$, the forward-looking and myopic evaluations
become identical. Either way confers great analytical tractability
and avoids introducing an explicit discount factor.}

After selling the patent, the entrepreneur exits, so only the immediate
cash prize matters. Then, the expected instantaneous payoff of a potential
innovator is given by: 
\begin{equation}
\mathbb{E}\left[V_{i}\right]=\frac{\Phi_{i}}{r}\pi_{i}(s^{**})-\frac{\alpha}{r}w_{i}.\label{eq:expectedprofits}
\end{equation}
We are assuming that incumbents do not engage in research. Since they
are already earning monopoly profits and entry into R\&D races is
free, incumbents have weaker incentives to conduct research activities.
This is due to the replacement effect pointed out by \citet{arrow1962economic}.\footnote{This assumption can be relaxed if we consider that quality leaders
have cost advantages in innovation.}

Successful innovators all draw from the same pool of public knowledge,
measured by the frontier quality grade across all varieties in both
regions:\footnote{Since all varieties are produced by quality leaders, we drop the quality
grade subscript hereinafter to ease the burden of notation unless
needed for clarity. We will also display or omit the explicit dependence
of state variables on time $t$ according to clarity of exposition.} 
\[
a_{\text{max}}(t)\equiv\sup_{s\in\mathcal{S}_{i},i\in\{1,2\}}a_{i}(s,t).
\]

\noindent We follow \citet[Chap. 12]{Aghion-Howitt-Book1998} and
consider a global spillover whereby all successful innovations push
the frontier quality $a_{\text{max}}$ forward. This also implies
that laggard firms are able to leapfrog firms with frontier quality
$a_{\text{max}}$. Formally, we have the following assumption. 
\begin{assumption}
\noindent Whenever a firm innovates at time $t$, the frontier quality
jumps by a factor of $\delta$: $a_{\text{max}}(t^{+})=\delta a_{\text{max}}(t^{-}).$
The innovator's quality becomes $a_{\text{max}}(t^{+})$, while the
qualities of all other varieties remain unchanged.
\end{assumption}
\noindent At any time $t$, there is a distribution of quality levels
$a_{i}(s)$ in region $i$ ranging from $\underline{a_{i}}\geq1$
to $a_{\text{max}}.$ For each new variety being produced by a new
firm at some time $t=t_{0}^{s},$ we have $m(s)=0$ so that $a_{i}(s)=1$
for all varieties $s\in\mathcal{S}_{i}$ at $t_{0}^{s}=0$. This also
implies that $m(s)=0$ for all varieties produce from $t=0$.

Let us define the relative quality of a firm as $a_{i}=\frac{a_{i}(s)}{a^{max}}\in[0,1]$.
We assume that $a_{i}$ is distributed according to a cumulative distribution
function $H_{i}(a),$ with support $[0,1]$. The average \emph{relative}
quality in region $i$ is given by 
\[
\bar{a}_{i}=\mathbb{E}\left[a|i\right]=\int_{0}^{1}ah_{i}(a)da,
\]
where $h_{i}(a)=H_{i}^{\prime}(a)$. Labour market clearing for researchers
implies that the number of varieties (and hence firms) is given by
$n_{i}=z_{i}/\alpha$. Therefore, the aggregate quality in region
$i$ is given by 
\begin{equation}
A_{i}=\frac{z_{i}}{c}\bar{a}_{i},\label{eq:aggregatequality}
\end{equation}
where $c\equiv\frac{\alpha}{a_{\text{max}}}$ is research effort per
variety adjusted by the frontier quality.

\noindent 




\subsection{Short-run equilibrium}

We can now characterize the short-run equilibrium at time $t$ by
summarizing the timing of events over an infinitesimal interval $[t,t+\mathrm{d}t).$\footnote{Equivalently, we could consider a short-run equilibrium over a discrete-time
one-period setup with probability of success $\Pi_{i}=1-e^{-\Phi_{i}\Delta t}$.
With a small hazard, we have $\Pi_{i}\approx\Phi_{i}$, making the
discrete and continuous time formulations approximately equivalent.} At $t$, the leading quality grade of any variety $s$ is $k$, inherited
from past innovations. Each outsider firm engages in research and
pays the wage flow $\frac{\alpha}{r}w_{i}\mathrm{d}t$ to researchers.
The success probability of innovation for a variety is $\Phi_{i}\mathrm{d}t$.
If no innovation arrives, the incumbent remains the quality leader
after $t+\mathrm{d}t$. If an innovation arrives, the successful entrepreneur
upgrades the quality of the chosen variety to $a(s^{**})=a_{\text{max}}(t^{+})=\delta a_{\text{max}}(t^{-})$,
immediately sells the patent for the lump-sum prize $V_{i}\equiv\pi_{i}(s^{**})$
to a newly formed production firm and exits.

To close the model, we assume that production firms are collectively
owned by their workers. Researchers are not residual claimants on
firm profits, ensuring that R\&D labor supply decisions remain governed
by the expected innovation prize. This setup preserves tractability
and supports a general equilibrium interpretation without requiring
forward-looking behavior.

Production firms' operating profits are distributed as dividends to
workers', in units of the numéraire good. Workers use their income
coming from wages and dividends for consumption and to finance patent
purchases when innovations occur. Residual gains, $R_{i}$, accrue
to the endowment $\bar{B}_{i}=\bar{B}+R_{i}$, where $\bar{B}>0$
is an initial endowment large enough to guarantee the consumption
of the numéraire. Given the quasi-linear utility in (\ref{eq:utility}),
dividends do not influence the workers' demand for manufactured goods.
As a result, dividends have no impact on production firms' operating
profits, which means they do not affect the wages paid to researchers. 

In the short-run equilibrium, production quality leaders earn positive
operating profits; free entry applies to R\&D (not production). With
free entry in R\&D, the expected payoff flow of entrepreneurs is zero
at the variety level. Using (\ref{eq:expectedprofits}), this yields
the following condition:\footnote{In \citet{melitz2008market}, zero expected profit applies to production
with heterogeneous productivities, yielding a cutoff and selection.
Introducing Melitz-Ottaviano-style heterogeneity here would add a
selection margin on top of R\&D; the prize becomes the expected operating
profit conditional on clearing the cutoff productivity level. This
would deliver selection/sorting while preserving our innovation-driven
frontier dynamics. We view this as a promising extension.} 
\begin{equation}
w_{i}=\frac{\Phi_{i}\pi_{i}(s^{**})}{\alpha},\label{eq:freeentry}
\end{equation}
where $\pi_{i}(s^{**})=a_{i}(s^{**})\tilde{\pi}_{i}$ is given by
$(\ref{eq:profitreduced})$. Using (\ref{eq:qualityzeroprofit}) in
(\ref{eq:freeentry}), this yields

\noindent
\begin{align*}
w_{i} & =\frac{\mu}{\sigma\alpha}\Phi_{i}a_{i}(s^{**})\left(\frac{\frac{\lambda}{2}+z_{i}}{A_{i}+\phi A_{j}}+\phi\frac{\frac{\lambda}{2}+z_{j}}{\phi A_{i}+A_{j}}\right).
\end{align*}
In terms of relative qualities, using (\ref{eq:aggregatequality})
and $c\equiv\frac{\alpha}{a_{\text{max}}}$, we get 
\begin{equation}
w_{i}=\frac{\mu}{\sigma}\Phi_{i}\delta\left(\frac{\frac{\lambda}{2}+z_{i}}{z_{i}\bar{a}_{i}+\phi z_{j}\bar{a}_{j}}+\phi\frac{\frac{\lambda}{2}+z_{j}}{\phi z_{i}\bar{a}_{i}+z_{j}\bar{a}_{j}}\right).\label{eq:wage equation}
\end{equation}

\noindent The price index in (\ref{eq:quality price index 3}) becomes:
\begin{equation}
P_{i}=\frac{\beta\sigma}{\sigma-1}\left(\frac{1}{c}\sum_{j}z_{j}\phi_{ij}\bar{a}_{j}\right)^{\frac{1}{1-\sigma}}.\label{eq:quality price index 2-1}
\end{equation}
Finally, using (\ref{eq:wage equation}), (\ref{eq:quality price index 2-1})
and (\ref{eq:indirect utility}), we get the indirect utility of a
researcher in region $i$: 
\begin{align}
v_{i}= & \frac{\mu}{\sigma}\Phi_{i}\delta\left(\frac{\frac{\lambda}{2}+z_{i}}{z_{i}\bar{a}_{i}+\phi z_{j}\bar{a}_{j}}+\phi\frac{\frac{\lambda}{2}+z_{j}}{\phi z_{i}\bar{a}_{i}+z_{j}\bar{a}_{j}}\right)+\dfrac{\mu}{\sigma-1}\ln\left(z_{i}\bar{a}_{i}+\phi z_{j}\bar{a}_{j}\right)+\eta,\label{eq:indirectutilityrelativequality}
\end{align}
where $\eta\equiv-\mu\left(\frac{\beta\sigma}{\sigma-1}\right)-\mu+\frac{\mu}{\sigma-1}\ln\left(\frac{1}{c}\right)+\bar{B}$.
We assume that $\bar{B}_{i}=\bar{B}$ for researchers in both regions.

\section{Knowledge diffusion and spillovers}

\subsection{The dynamics of innovation}

Innovation and knowledge diffusion are assumed to operate on a much
faster timescale than the long-run\emph{ }structural forces that change
the spatial distribution of researchers $z$ and firms (e.g., migration,
entry/exit, plant relocation). We have the following assumption. \begin{assumption}[Time--scale
separation] \label{assu:Time=00003D002013scale-separation:-over}
Over the horizon relevant for quality adjustments, the researchers'
spatial distribution $\Vtz=(z_{1},z_{2})$ is fixed. \end{assumption}
This separates the fast adjustment of quality levels from the slow
evolution of $z$ in the long-run. Accordingly, we will (i) analyze
innovation dynamics for a given $z$, and (ii) study the long-run
evolution of $z$ assuming quality levels have rapidly adjusted to
their conditional steady-states. This time-scale separation is akin
to the fast-equilibrium/slow-migration structure in many standard
core--periphery models (\citealp[see e.g.][]{fujita1999spatial}).

The frontier quality evolves according to the following differential
equation: 
\begin{equation}
\dot{a}_{\text{max}}(t)=(\ln\delta)\left(n_{1}\Phi_{1}+n_{2}\Phi_{2}\right)a_{\text{max}}(t).\label{eq:frontierdynamics}
\end{equation}

\noindent When a hit occurs for a variety with relative quality $a\in[0,1]$
its relative quality jumps to $1$. The instantaneous contribution
of such events to the change in the regional average $\bar{a}_{i}$
is 
\[
\Phi_{i}\int_{0}^{1}(1-a)h(a)da=\Phi_{i}(1-\bar{a}_{i}).
\]

\noindent Whenever the global frontier is raised, all relative qualities
are scaled down by the same factor $\delta^{-1}.$ If the frontier
rises at an instantaneous growth rate $g\equiv\frac{\dot{a}_{\text{max}}}{a_{\text{max}}}$,
then, to first order, every relative quality, and therefore the average,
drifts down at rate $g$. The frontier drift is given by $-g\bar{a}_{i}.$
Let $\Phi_{i}=\Phi_{i}(\bar{a}_{i},\bar{a}_{j})$ denote the per--variety
arrival hazard in region $i$ at the current state $(\bar{a}_{i},\bar{a}_{j})\in\mathcal{H}\equiv(0,1)^{2}$.
The dynamics of average relative quality in region $i$ are given
by the following differential equation: 
\begin{equation}
\frac{d\bar{a}_{i}}{dt}=\Phi_{i}(1-\bar{a}_{i})-g\bar{a}_{i}+D_{i}(\bar{a}_{i},\bar{a}_{j}),\label{eq:relativequalitymotion}
\end{equation}
where 
\[
D_{i}(\bar{a}_{i},\bar{a}_{j})=-\frac{1}{2}(\Phi_{i}-\Phi_{j})(1-\bar{a}_{i}),
\]
is a minimal ``size-offset'' diffusion/congestion term that captures
knowledge diffusion in reduced form and eliminates any size/home-bias
embedded in $\Phi_{i}$.\footnote{See \citet{BENHABIB2005935} for mechanisms of international knowledge
diffusion/catching-up.} Thus, (\ref{eq:relativequalitymotion}) becomes simply: 
\begin{equation}
\frac{d\bar{a}_{i}}{dt}=\frac{1}{2}\left(\Phi_{1}+\Phi_{2}\right)(1-\bar{a}_{i})-g\bar{a}_{i}.\label{eq:relativequalitymotionfinal}
\end{equation}
The term $\frac{1}{2}(\Phi_{1}+\Phi_{2})$ can be interpreted as the
effective reset hazard at which varieties jump back to the frontier
($a_{i}=1$), equalized across regions by knowledge diffusion through
the term $D_{i}$.

Mathematically, interior invariance follows directly from (\ref{eq:relativequalitymotionfinal}),
because, at $a_{i}=0$, the drift is strictly positive and, at $a_{i}=1$,
it is strictly negative. Therefore, boundary points are not absorbing
and cannot be reached from interior initial conditions. This is why
we can restrict the state space to $\mathcal{H}$ and ignore boundary
points. Assumption \ref{assu:Time=00003D002013scale-separation:-over}
then implies that the system in (\ref{eq:relativequalitymotionfinal})
converges quickly to a conditional steady state $(\bar{a}_{1},\bar{a}_{2})=(\bar{a},\bar{a})\in\mathcal{H}$
before the long-run is reached. Afterwards, slow movements in $z$
shift this conditional steady state over the long-run.

The regional instantaneous Poisson arrival rate of successful innovations
(hazard rate) for any variety $s\in\mathcal{S}_{i}$ is given by 
\[
\Phi_{i}=F_{i}(A_{i},A_{j};\boldsymbol{b})G_{i}\left(\frac{\bar{a}_{j}}{\bar{a}_{i}}\right)\xi\left(c\right),\quad\xi^{\prime}>0,\quad\xi^{\prime\prime}\leq0,\quad G^{\prime}>0,
\]
where $F_{i}:\mathcal{D}\subseteq\mathbb{R}_{+}^{2}\to\mathbb{R}_{+}$
is called the \emph{regional spillover function, }assumed to be continuously
differentiable, and with the property $F_{i}(\lambda A_{i},\lambda A_{j})=\lambda F_{i}(A_{i},A_{j})$
for all $\lambda>0$ and for all $(A_{i},A_{j})\in\mathcal{D}$; $\boldsymbol{b}$
is a vector of parameters that \emph{does} \emph{not }include the
freeness of trade $\phi$ (for now). Following \citet{Li-AER2003},
spillovers scale with global rather than relative aggregate quality,
which is observed by all firms. We follow \citet{segerstrom1998endogenous}
and \citet{Howitt-JPE1999} and introduce a ``complexity term''
$\xi:\mathbb{R}_{+}\to\mathbb{R}_{+}$, which is a function of $c$,
reflecting the fact that innovation success is decreasing in the complexity
of each product, as measured by $a_{\text{max}}$, but increasing
in the number of researchers per firm $\alpha$. Finally, $G_{i}:\mathbb{R}_{+}\to\mathbb{R}_{+}$
is the \emph{catching-up function, }with $G_{i}(1)=1$\emph{. }The
regional hazard rate $\Phi_{i}$ rises when region $i$ is a laggard:
larger technology gaps create more profitable imitation/adoption opportunities
and stronger R\&D incentives. That is exactly what $G_{i}(\cdot)$
captures: an advantage in backwardness (\citealp{BENHABIB2005935}).

For simplicity, assume that the complexity term is given by $\xi\left(c\right)=\frac{1}{\delta}\frac{\alpha}{a_{\text{max}}}.$\footnote{Constant returns to scale with respect to research effort $\alpha$
may be hard to justify empirically but substantially simplifies the
analysis.} We have the following result. 
\begin{thm}
\label{thm:convergence}In the long-run, regional relative qualities
are given by 
\begin{equation}
\bar{a}_{1}=\bar{a}_{2}=\bar{a}\equiv\frac{F_{1}+F_{2}}{F_{1}+F_{2}+\frac{2}{\alpha\delta}\ln\delta\left[z_{1}F_{1}+z_{2}F_{2}\right]},\label{eq:relativequalityss}
\end{equation}
where $F_{i}\equiv F_{i}(z_{i},z_{j})$. The frontier growth rate
is given by 
\begin{equation}
g=\frac{\ln\delta}{\alpha\delta}\bar{a}\left(z_{1}F_{1}+z_{2}F_{2}\right).\label{eq:frontiergrowthrate}
\end{equation}
\end{thm}
\begin{proof}
See Appendix A. 
\end{proof}
Theorem \ref{thm:convergence} establishes that the step-size $\delta$
decreases the average relative quality, while total research effort
per firm $\alpha$ increases it. Crucially though, the average quality
level depends on the spatial distribution $\Vtz$ and how it affects
the innovation rate through the spillover function $F_{i}.$

The innovation rate for a variety $s$ in region $i$ becomes 
\begin{equation}
\Phi_{i}=\frac{\bar{a}}{\delta}F_{i}(z_{i},z_{j};\boldsymbol{b}),\label{eq:phiFi}
\end{equation}
where $\bar{a}$ is given by (\ref{eq:relativequalityss}) and the
frontier growth rate $g$ is given by (\ref{eq:frontiergrowthrate}).
The cross-sectional stationary distribution of relative quality is
given by the power-law $H_{i}(a)=a^{\theta}$ on $[0,1]$, with $\theta=\frac{\upsilon}{g}$
and $\upsilon=\frac{\Phi_{1}+\Phi_{2}}{2}$.\footnote{To see this, note that, between resets, $a(t)=e^{-gT}$, where $T\sim\text{Exp}(\rho)$
is the time since the last reset. We have $H(a)=\Pr(e^{-gT}\leq a)=\Pr(T\geq e^{\frac{1}{g}\ln\left(\frac{1}{a}\right)})=e^{-\frac{\upsilon}{g}\ln\left(\frac{1}{a}\right)}=a^{\frac{\upsilon}{g}},$
with $\upsilon=\frac{\Phi_{1}+\Phi_{2}}{2}$.} The effective reset hazard $\upsilon$ implies that knowledge created
anywhere diffuses quickly enough that, on average, each region benefits
from half of total innovation arrivals. 

Note how, through $F_{i}(\Vtz)$, the knowledge-stock-based view of
regional spillovers is equivalent to a researcher-pool-base view.
Thus, spillovers can be expressed in terms of the (spatially weighted)
distribution of researchers across regions.

\subsection{Regional spillovers}

We now introduce a specific functional form for the spillover function.
Specifically, we set 
\begin{align}
 & F_{i}(\Vtz)=\gamma\left[bz_{i}+(1-b)z_{j}\right]=F(z_{i}),
\end{align}
where $b\in(0,1)$ is the weight of local knowledge contribution on
the innovation rate, capturing the relative importance of intra- versus
inter-regional knowledge spillovers, as discussed in Section~\ref{subsec:spillover-mechanisms},
and $\gamma>0$ captures the overall intensity of regional spillovers.
When $b>\tfrac{1}{2}$, domestic knowledge or effort plays a dominant
role. When $b<\tfrac{1}{2}$, innovation is driven more by external,
inter-regional sources. In what follows, we shall refer to $b$ as
the \emph{relative intensity/weight of local spillovers. }We define
\begin{align}
 & F(z)=\gamma\left[bz+(1-b)(1-z)\right] & z\in[0,1],\label{eq:regionalspillover}
\end{align}
so that $F_{1}(\Vtz)=F(z_{1})$ and $F_{2}(\Vtz)=F(1-z_{1})$.

The rationale for this specification is as follows. Analogously to
the interpretation of \citet{Berliant-Fujita-RSUE2012}, we implicitly
assume that each region holds its own set of ideas (or culture). Therefore,
production of knowledge depends on the amount of ``within region''
interaction among researchers, but also on the interaction with researchers
hailing from a different region. A common mechanism for the interaction
between researchers from different regions is the mobility of individuals
and the trade or transfer of goods, which, in one way or another,
carry production related knowledge with them \citep{Doring01052006}.
The interaction between agents both within a region and across different
regions generates knowledge spillovers that foster innovations. The
spatial nature of regional spillovers depends on the composition and
types of knowledge that mainly characterizes the innovation processes
in a region (cf.\@ Section 3).

If the degree of sectoral specialization is relatively homogeneous
within regions, then regional spillovers depend directly on the spatial
distribution of agents. In turn, the impact of the spatial distribution
on innovation depends on whether knowledge spillovers are more localized
or more globalized. A higher $b$ means that spillovers are stronger
with a higher pool of researchers living in the same region. If $b>\frac{1}{2}$,
$F'(z)>0$, and we say that local spillovers dominate. If $b<\frac{1}{2}$,
then $F'(z)<0$, and we say that global spillovers dominate. Thus,
when local (global) spillovers dominate, knowledge spillovers constitute
an additional agglomeration (dispersion) force.

Higher global spillovers may be associated with a leapfrogging mechanism
whereby more industrialized regions fail to benefit from breakthroughs
in innovation because they are too reliant on ``old'', more incremental
technologies, as argued by \citet{brezis1993}. Since dispersive forces
such as congestion in innovation are stronger in the larger region,
this creates opportunities for firms to locate in the smaller region,
which eventually catches up with the other region and promotes a more
balanced spatial distribution. Knowledge spillovers induce a dispersion
force in this case. On the other hand, a high local spillovers are
associated with the fact that innovative production facilities with
a high intensity of tacit knowledge and frequency of direct communication
tend to cluster together as they benefit from a higher proximity between
agents \citep{Doring01052006}.\footnote{The need for high geographical proximity might be mitigated by a higher
level of economic integration between regions as it reduces the distance
decay effect, making knowledge spillovers more effective across regions.} Although the existing empirical literature supports localized spillovers,
the ongoing globalization of innovation is bound to promote the inter-regional
transfer and diffusion of knowledge and thus calls for a theoretical
and empirical refinement of the spatial nature of innovation \citep{carlino2015agglomeration,miguelez2018relatedness}.

\subsection{Indirect utility under the regional spillover}

We first provide a convenient representation of the indirect utility
based on one variable. Since $z_{1}$ completely characterizes $\Vtz=(z_{1},z_{2})=(z_{1},1-z_{1})$,
we use $z\in[0,1]$ as the state variable, so that $z_{1}=z$ and
$z_{2}=1-z$.

The innovation rate in region $1$ is given by: 
\begin{equation}
\Phi_{1}(z)=\frac{\gamma}{\delta}\left[bz+(1-b)(1-z)\right]\bar{a}.\label{eq:ratesimplecase}
\end{equation}

\noindent The steady-state average relative quality in (\ref{eq:relativequalityss})
becomes 
\[
\bar{a}(z)=\frac{1}{1+\frac{2}{\alpha\delta}(\ln\delta)\left[b(1-2z)^{2}+2(1-z)z\right]}.
\]
It can be shown that $\frac{d\bar{a}}{dz}<(>)0$ for $z>(<)\frac{1}{2}$
when $b>\frac{1}{2}$. Therefore, when local spillovers dominate,
the average distance to the frontier is maximized when $z=\frac{1}{2}$
and is minimized at $z\in\{0,1\}$. If global spillovers dominate,
we have the converse case. As for the growth rate of the frontier
quality, $g$ in (\ref{eq:frontiergrowthrate}), it becomes 
\[
g=\omega\left[1-\frac{1}{b(1-2z)^{2}-2(z-1)z+1}\right],
\]
where $\omega=\frac{\ln\delta}{\delta\alpha}$. It is easily observable
that $\frac{dg}{dz}>(<)0$ for $z>(<)\tfrac{1}{2}$ when $b>\tfrac{1}{2}$.
Hence, the growth rate of $a_{\text{max}}$ is maximized for $z\in\{0,1\}$
and minimized when $z=\tfrac{1}{2}$ if spillovers are mainly local.
The converse case is true when $b<\tfrac{1}{2}.$

From Theorem \ref{thm:convergence}, and plugging (\ref{eq:ratesimplecase})
into (\ref{eq:indirectutilityrelativequality}) , the long-run indirect
utility in region $1$ is 
\begin{align}
v_{1}(z)= & \frac{\mu}{\sigma}F(z)\left(\dfrac{\frac{\lambda}{2}+z}{z+\phi(1-z)}+\phi\dfrac{\frac{\lambda}{2}+1-z}{\phi z+1-z}\right)+\dfrac{\mu}{\sigma-1}\ln\left[z+\phi(1-z)\right]+\eta,\label{eq:indirectutilityregion1}
\end{align}
where $F(z)$ is given by \eqref{eq:regionalspillover} and 
\[
\eta=-\mu\left(\frac{\beta\sigma}{\sigma-1}\right)-\mu+\frac{\mu}{\sigma-1}\ln\left(\frac{\bar{a}}{c}\right)+\bar{B}.
\]

\noindent If local spillovers dominate and region $1$ is larger,
an increase in $z$ amplifies the nominal wages through $F(z)$ but
also increases the cost-of-living through $\bar{a}$. The converse
is true when spillovers are mainly global and region $1$ is smaller.
However, since average relative qualities are equalized across regions,
the difference between regional indirect utilities only depend on
$F_{1}$ and $F_{2}$. This setting is enough to provide detailed
and interesting insights throughout the next sections without sacrificing
analytical tractability.

\section{Long-run equilibria: researchers and firms}

\label{sec_lrequ}

In the long-run, researchers are free to migrate between regions.
In doing so, they choose the region that offers them the highest indirect
utility. The long-run spatial distribution thus depends on the utility
differential: 
\begin{equation}
\Delta v(z)=v_{1}(z)-v_{2}(z).\label{eq:utility differential}
\end{equation}
Under Assumption \ref{assu:Time=00003D002013scale-separation:-over},
quality adjusts ``instantaneously'' relative to migration, so along
migration paths we evaluate relative quality levels at its conditional
steady state. For each $z$, the quality block has a unique explicit
conditional steady state, symmetric across regions, $\bar{a}_{1}=\bar{a}_{2}=\bar{a}(z)$,
given by (\ref{eq:relativequalityss}). Let the migration law of researchers
be 
\[
\dot{z}=\kappa\Delta v\big(z;\bar{a}_{1},\bar{a}_{2}\big),\qquad\kappa>0.
\]
Evaluating at the conditional steady state yields the reduced form
\begin{equation}
\dot{z}=\kappa\Delta v\big(z;\bar{a}(z),\bar{a}(z)\big)\equiv\kappa\Delta v(z).\label{eq:migrationdynamics}
\end{equation}
As discussed, the term $\eta$ in (\ref{eq:indirectutilityregion1})
cancels when subtracting utilities, the quality level $\bar{a}$ does
not affect the indirect utility differential $\Delta v$. We thus
neutralize the effect of quality and preserve the structural symmetry
across regions. Asymmetries in innovation stem uniquely from the knowledge
spillover functions $F_{1}$ and $F_{2}$. As a result, the spatial
outcomes of our model are solely determined by pecuniary factors and
by the regional spillover functions $F_{1}=F(z)$ and $F_{2}=F(1-z)$.
In other words, the long-run spatial distribution of researchers and
firm is a result of endogenous agglomeration mechanisms operating
through market factors and the economic geography of knowledge spillovers.

We follow \citet{Castro_2021} in the characterization of equilibria
and their stability. There are two kinds of long-run equilibria which
should be dealt with separately. 
\begin{enumerate}
\item \emph{Agglomeration }of all researchers in a single region $z^{*}=\left\{ 0,1\right\} $
is an equilibrium if and only if $\Delta v(1)\geq0,$ or, equivalently,
$\Delta v(0)\leq0$. 
\item \emph{Dispersion} of researchers $z^{*}\in\left(0,1\right)$ is an
equilibrium if and only if $\Delta v(z^{*})=0$. If $z^{*}=\frac{1}{2}$
it corresponds to \emph{symmetric dispersion.} Otherwise, it is called
\emph{asymmetric}. 
\end{enumerate}
Equilibria are stable if, after a perturbation such that $z=z^{*}\pm\epsilon$,
with $\epsilon>0$ small enough, the utility differential $\Delta v(z)$
becomes such that agents go back to their place of origin, i.e., $z=z^{*}$.
A sufficient condition for stability of agglomeration is $\Delta v(1)>0$
(or $\Delta v(0)<0$). A sufficient condition for the stability of
dispersion is that $\Delta v^{\prime}(z^{*})<0$. When equilibria
are \emph{regular} (resp. $\Delta v(1)\neq0$ and $\Delta v^{\prime}(z^{*})\neq0$),
these conditions are also necessary.\footnote{In models that are well-behaved, the non-existence of irregular long-run
equilibria holds in a full measure subset of a suitably defined parameter
space.}

\subsection{Existence and multiplicity}

Our first result regards the multiplicity of long-run equilibria.
Given symmetry across regions, we focus on the case whereby region
$1$ is either the same size or is larger than region $2$, i.e.,
$z\in\left[\frac{1}{2},1\right]$.

Symmetric dispersion $z^{*}=\frac{1}{2}$ is called an \emph{invariant
pattern}, because it is a long-run equilibrium for the entire parameter
range \citep{aizawa2020break5}. Next, we have the following result
regarding possible equilibria for $z\in\left(\frac{1}{2},1\right].$ 
\begin{prop}
There are at most two equilibria for $z\in\left(\frac{1}{2},1\right].$ 
\end{prop}
\begin{proof}
See Appendix B.1. 
\end{proof}
We can be more precise regarding the existence of asymmetric dispersion
equilibria with the following result. 
\begin{prop}
An asymmetric dispersion equilibrium $z\equiv z^{*}\in\left(\frac{1}{2},1\right]$
exists if $b\in\left(\max\left\{ 0,\tilde{b}\right\} ,\hat{b}\right)$,
where:{\small{} 
\[
\tilde{b}\equiv\frac{\gamma(\sigma-1)(2z-1)\left[(z-1)z\left(\phi^{2}-1\right)+\phi^{2}\right]+\sigma\left[z(\phi-1)+1\right]\left[z(\phi-1)-\phi\right]\ln\left[\frac{z(\phi-1)+1}{z(1-\phi)+\phi)}\right]}{\gamma(\sigma-1)(2z-1)(\phi+1)\left[2(z-1)z(\phi-1)+\phi\right]},
\]
}and 
\[
\hat{b}\equiv\frac{1+\phi^{2}}{(1+\phi)^{2}}.
\]
\end{prop}
\begin{proof}
See Appendix B.2 
\end{proof}

\begin{figure}[!h]
\begin{centering}
\subfloat[Stable symmetric dispersion: $\phi=0.1$.]{\begin{centering}
\includegraphics[scale=0.6]{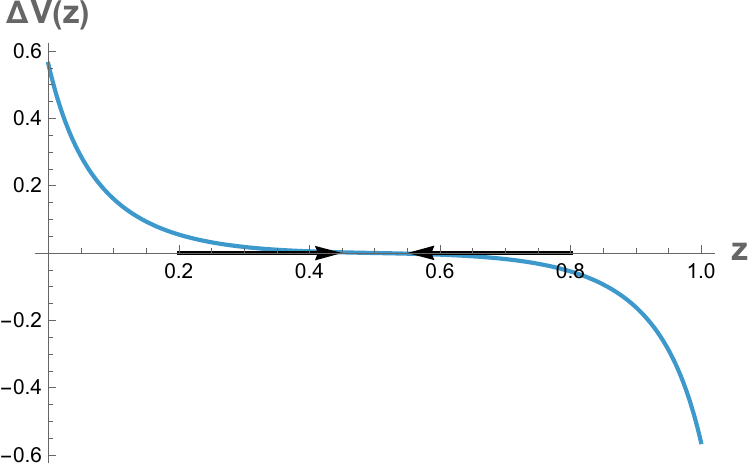} 
\par\end{centering}
\label{subfigure1a}

}\subfloat[Stable asymmetric dispersion: $\phi=0.3$.]{\begin{centering}
\includegraphics[scale=0.6]{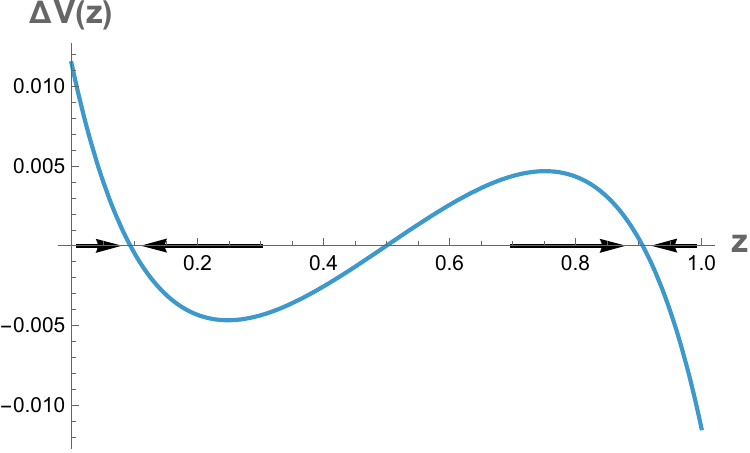} 
\par\end{centering}
\label{subfigure1b}

}
\par\end{centering}
\begin{centering}
\subfloat[Stable agglomeration: $\phi=0.38$]{\begin{centering}
\includegraphics[scale=0.6]{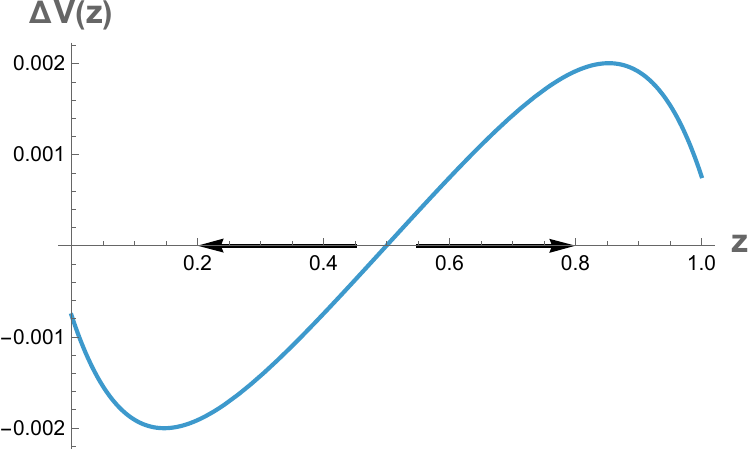} 
\par\end{centering}
\label{subfigure1c}

}\subfloat[Stable dispersions: $\phi=0.4$.]{\begin{centering}
\includegraphics[scale=0.6]{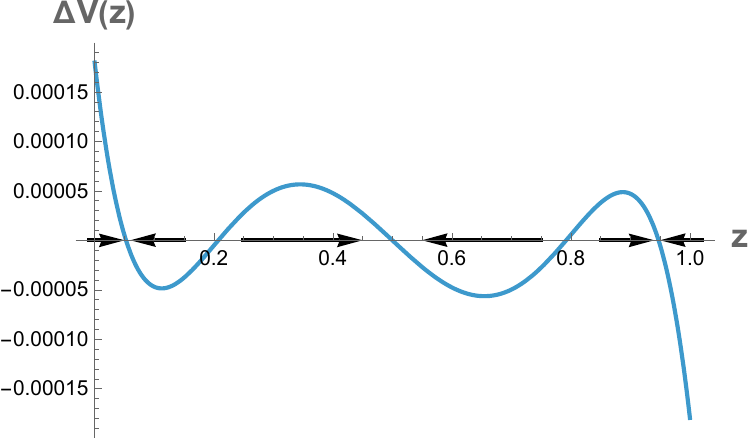} 
\par\end{centering}
\label{subfigure1d}

}
\par\end{centering}
\begin{centering}
\subfloat[Stable symmetric dispersion: $\phi=0.8$.]{\begin{centering}
\includegraphics[scale=0.6]{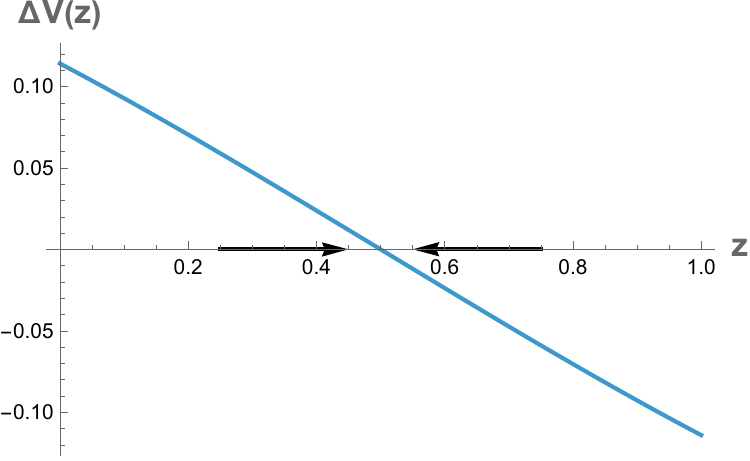} 
\par\end{centering}
\label{subfigure1e}

}
\par\end{centering}
\caption{Long-run equilibria and their stability as $\phi$ increases.}
\label{fig:Figure1} 
\end{figure}

Figure~\ref{fig:Figure1} shows one scenario with five qualitatively
different cases, with varying freeness of trade $\phi$, for the case
when inter-regional spillovers dominate, $b\in(0,\frac{1}{2})$, regarding
existence of the model's long-run spatial distribution for the parameter
values $\left(\lambda,\gamma,\sigma,b\right)=\left(2,1,5,0.342\right)$.
These examples exhaust all qualitative possibilities. As a prelude
to the forthcoming Section, Figure~\ref{fig:Figure1} also numerically
depicts the local stability of each equilibrium, which is to be analysed
analytically in greater detail in Section 6.2.

In Figure~\ref{subfigure1a}, only symmetric dispersion exists and
is stable for a very small $\phi$. For a higher trade freeness we
have one stable asymmetric dispersion for $z\in\left(\frac{1}{2},1\right]$
as portrayed in Figure~\ref{subfigure1b}. For a greater $\phi$,
Figure~\ref{subfigure1c} shows that the asymmetric dispersion equilibrium
disappears and symmetric dispersion becomes unstable, whereas agglomeration
becomes stable. For an even greater $\phi$, Figure~\ref{subfigure1d}
illustrates an example of two long-run dispersion equilibria $z^{*}$
for $z\in\left(\frac{1}{2},1\right]$ and symmetric dispersion $z=\frac{1}{2}$,
whereby we can observe that both symmetric dispersion and the more
agglomerated asymmetric dispersion equilibrium are locally stable,
whereas the less agglomerated asymmetric dispersion equilibrium is
unstable. The economy re-disperses and agglomeration does not exist
in this particular case. However, as the trade freeness increases
further, symmetric dispersion remains stable and all other equilibria
disappear, as demonstrated by Figure\ \ref{subfigure1e}.

When inter-regional spillovers dominate, $b\in(0,\frac{1}{2})$, the
model accounts for a ``bell-shaped'' relationship between economic
integration and agglomeration (\citealp{Fujita2013}), whereby firms
are initially dispersed, then start to agglomerate in a single region
as the trade freeness increases, but then find it worthwhile to relocate
to the peripheral regions in order to benefit from higher expected
profits due to the sizeable pool of researchers in the core which
increases the chance of innovation in the periphery. In other words,
when inter-regional spillovers dominate, they constitute an additional
dispersion force whose strength becomes relatively higher as economic
integration increases and leads to the vanishing of agglomeration
forces.

The case when intra-regional spillovers dominate, $b\in(\frac{1}{2},1)$,
is much less diversified and can be accounted for resorting to a subset
of the pictures from Figure \ref{fig:Figure1}, but with $b>\tfrac{1}{2}$
and different values of $\phi$ (although still increasing). The history
as economic integration increases is as follows. For a very low trade
freeness, symmetric dispersion is the only stable equilibrium as in
Figure~\ref{subfigure1a}. For an intermediate value of $\phi$,
one asymmetric dispersion equilibrium arises which is the only stable
one and becomes more asymmetric as $\phi$ increases further. This
is akin to the picture in Figure~\ref{subfigure1b}. Finally, the
asymmetric dispersion equilibrium gives rise to stable full agglomeration
in one single region once $\phi$ becomes very high. This is illustrated
in Figure~\ref{subfigure1c}. In other words, when intra-regional
interaction is relatively more important, knowledge spillovers are
more localized, and thus constitute an additional agglomeration force.

In the forthcoming Sections, we will analytically and numerically
study in greater detail the local stability of the spatial distributions
and the qualitative change in the model's structure as economic integration
increases.

\subsection{Stability}

\subsubsection{Agglomeration}

Regarding agglomeration, using (\ref{eq:indirectutilityregion1})
in (\ref{eq:utility differential}), we have that it is stable if:
\[
\Omega\equiv\frac{\gamma\left[(b-1)(\lambda+2)\phi^{2}+2b(\lambda+1)\phi+(b-1)\lambda\right]}{2\sigma\phi}-\frac{\ln\phi}{\sigma-1}>0.
\]
The second term is positive. Hence, agglomeration is always stable
if the first term is also positive: 
\[
b>b_{s}\equiv\frac{(\lambda+2)\phi^{2}+\lambda}{(\phi+1)\left[(\lambda+2)\phi+\lambda\right]}.
\]
It is easy to check that $b_{s}<\tfrac{1}{2}$ if $\phi\in\left(\frac{\lambda}{\lambda+2},1\right)$,
which means that, if $\phi\in\left(\frac{\lambda}{\lambda+2},1\right)$
and $b>\frac{1}{2}$, agglomeration is stable. In any case, we can
conclude that agglomeration is always stable if the intensity of local
spillovers is extremely high.

Let us now define as \emph{sustain point} \citep{fujita1999spatial},
a value of $\phi$ such that $\mathcal{S}(\phi)=0.$ We have the following
result relating the freeness of trade and the relative weight of local
spillovers. 
\begin{prop}
If $b<\tfrac{1}{2}$, there exist at most two sustain points, $\phi_{s1}$
and $\phi_{s2},$ and agglomeration is unstable for $\phi\in\left(0,\phi_{1s}\right)\cup\left(\phi_{2s},1\right)$
and stable for $\phi\in(\phi_{1s},\phi_{2s})$. If $b>\frac{1}{2}$,
there exists a unique sustain point $\phi_{s1}$ and agglomeration
is unstable for $\phi\in(0,\phi_{1s})$ and stable if $\phi\in\left(\phi_{1s},1\right)$. 
\end{prop}
\label{prop3} 
\begin{proof}
See Appendix B.3. 
\end{proof}
The result in Proposition 3 suggests that an intermediate level of
economic integration favours agglomeration if the interaction with
foreign researchers is relatively more important for innovation. By
contrast, if the intra-region interaction of researchers is more important,
agglomeration is possible when the freeness of trade is high enough.

\subsubsection{Symmetric dispersion}

Regarding symmetric dispersion $z^{*}=\frac{1}{2}$, using (\ref{eq:indirectutilityregion1})
in (\ref{eq:utility differential}) we can say that it is stable if:
\begin{equation}
\mathcal{B}\equiv\gamma(\sigma-1)\left[2b(\lambda+1)(\phi+1)^{2}-(2\lambda+3)\phi^{2}-2\lambda-1\right]+2\sigma\left(1-\phi^{2}\right)<0.\label{eq:stability symmetric dispersion}
\end{equation}
In fact, it is always unstable if the first term is positive, i.e.
if: 
\[
b>b_{d}\equiv\frac{(2\lambda+3)\phi^{2}+2\lambda+1}{2(\lambda+1)(\phi+1)^{2}}.
\]

\noindent This means that if weight of local contribution is prohibitively
high, symmetric dispersion is surely unstable.

We can observe that $\mathcal{B}$ in (\ref{eq:stability symmetric dispersion})
is a second degree polynomial in $\phi$ with at most two zeros, i.e.,\emph{
break points }$\phi_{b1}$ and $\phi_{b2}$, with $\phi_{b1}<\phi_{b2}$
and has a negative leading coefficient. Therefore, if both break points
exist, we have that symmetric dispersion is stable for $\phi\in(0,\phi_{b1})\cup(\phi_{b2},1)$
and unstable for $\phi\in\left(\phi_{b1},\phi_{b2}\right)$. The expressions
for the break points, along with the conditions for their existence,
are provided in Appendix B.4. We have the following result. 
\begin{prop}
For $b>\frac{1}{2},$ symmetric dispersion is stable for $\phi\in(0,\phi_{b1})$
and unstable for $\phi\in(\phi_{b1},1)$, provided that $b$ is not
too high (otherwise it is always unstable). If $b<\frac{1}{2},$ symmetric
dispersion is stable for $\phi\in(0,\phi_{b1})\cup(\phi_{b2},1)$
and unstable for $\phi\in(\phi_{b1},\phi_{b2}),$ provided that $b$
is not too low (otherwise it is always stable). 
\end{prop}
\begin{proof}
See Appendix B.4. 
\end{proof}
This means that, when global spillovers dominate, our model accounts
for the possibility of initial agglomeration as trade integration
increases from a low level, and (complete) re-dispersion once trade
integration becomes high enough.

\subsubsection{Asymmetric dispersion}

Although we cannot find an explicit stability condition for any asymmetric
dispersion equilibrium $z^{*}\in\left(\frac{1}{2},1\right)$, we can
use equation (\ref{eq:equilibriumcondition}) that solves the equilibrium
condition $\Delta v=0$ given implicitly by $\lambda=\lambda^{*}(z)$
in the proof of Proposition 2 (Appendix B.2).\footnote{The same approach was adopted e.g. by \citet{gaspar2018agglomeration},
\citet{Gaspar-et-al-RSUE2021} and \citet{saraiva2025disentanglement}.} Then the stability condition of an asymmetric dispersion equilibrium
is given by: 
\[
\left.\frac{d\Delta v}{dz}(z^{*})\right|_{\lambda=\lambda^{*}(z)}<0.
\]

\noindent Specifically, using (\ref{eq:indirectutilityregion1}) and
differentiating (\ref{eq:utility differential}) with respect to $z$,
and evaluating at (\ref{eq:equilibriumcondition}), we get that an
asymmetric equilibrium $z^{*}\in\left(\frac{1}{2},1\right)$ is stable
if $\lambda^{*}(z)>0$ and:

\begin{align}
\mathcal{G}\equiv & (2z-1)\left(\phi^{2}-1\right)\left[(2b-1)\gamma(\sigma-1)(1-2z)^{2}-\sigma\right]+\nonumber \\
 & +\sigma\left[2z^{2}(\phi-1)^{2}-2z(\phi-1)^{2}+\phi^{2}+1\right]\ln\left[\frac{z(\phi-1)+1}{z(1-\phi)+\phi}\right]<0.\label{eq:stabasdisp}
\end{align}

\noindent We have the following result. 
\begin{prop}
If $b<\frac{1}{2}$ an asymmetric equilibrium $z^{*}\in\left(\frac{1}{2},1\right)$
is stable for a high enough weight of local contribution. If $b>\frac{1}{2}$,
an asymmetric equilibrium is always stable when it exists. 
\end{prop}
\begin{proof}
See Appendix B.5. 
\end{proof}
\begin{figure}[!h]
\begin{centering}
\includegraphics[scale=0.75]{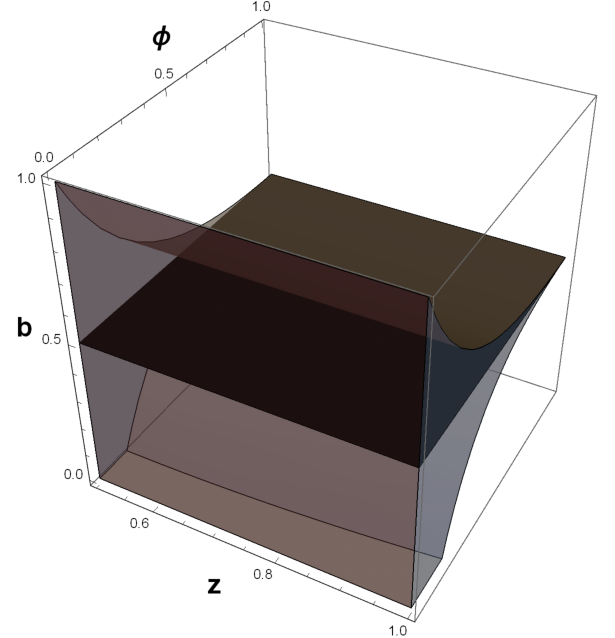} 
\par\end{centering}
\caption{Stability of asymmetric dispersion. The less opaque surface corresponds
to $\mathcal{G}<0$ and $\lambda^{*}(z)>0$ in $(z,\phi,b)$--space
for $\sigma=8$ and $\gamma=1$. The black plane corresponds to $b=\frac{1}{2}.$ }
\label{asymmdispersionfig} 
\end{figure}

Figure~\ref{asymmdispersionfig} illustrates Proposition 5 by setting
$\sigma=8$ and $\gamma=1$ and plotting the surface corresponding
to $\mathcal{A}=\left\{ (z,\phi,b):\mathcal{G}<0\cap\lambda^{*}(z)>0\right\} $.
For $b<\frac{1}{2}$, an asymmetric equilibrium may exist that is
not stable, and a higher $b$ favours its stability. If $b>\frac{1}{2}$,
an asymmetric equilibrium is always stable when it exists, but its
existence seems to be favoured by a lower $b$. In other words, an
asymmetric equilibrium exists and is stable when $b$ is close enough
to $\frac{1}{2}$.

Regarding $\phi$, a higher freeness of trade seems to disfavour the
stability of asymmetric dispersion.

\section{The impact of economic integration}

It is common in economic geography to study the qualitative change
of the spatial economy as economic integration increases. Accordingly,
we investigate the existence of bifurcations in our model and employ
the freeness of trade, $\phi$, as the bifurcation parameter. First,
suppose a break point exists, i.e., symmetric dispersion interchanges
stability for some value of the freeness of trade. We have the following
analytical result. 
\begin{prop}
At the break point $\phi=\phi_{b1}$, the symmetric dispersion undergoes
a supercritical pitchfork bifurcation. At the break point $\phi=\phi_{b2}$,
the symmetric dispersion undergoes a pitchfork bifurcation, which
may be subcritical, supercritical or degenerate.  
\end{prop}
\begin{proof}
See Appendix B.6.  
\end{proof}
Proposition 6 asserts that a curve of asymmetric equilibria branches
from symmetric dispersion at the break point, when it exists. The
stability of the branch depends on the criticality of the bifurcation.
If the pitchfork is subcritical, the asymmetric equilibria are unstable.
If it is supercritical, the asymmetric equilibria are stable. When
the break point $\phi_{b1}$ exists, a curve of stable asymmetric
equilibria branches from symmetric dispersion at $\phi=\phi_{b1}$
and lies to its right.

We now look at some bifurcation diagrams to illustrate the possible
transition dynamics between different equilibria as trade integration
smoothly increases. To provide a complete gallery, we depict 6 qualitatively
different scenarios, keeping most parameter values constant (except
for the sixth scenario) and varying $b$, thus placing emphasis on
changes in the intensity of local spillovers. These scenarios exhaust
all the different qualitative possibilities. This can be shown through
the combination of the analysis performed in Section 6, with several
illustrations, including contour plots showing the stability of different
equilibria in $(\phi,b)$-space (see Figure~\ref{phibspace}), under
a very wide range of parameter values.\footnote{Additionally, see Figure~\ref{phibspace}, the transition between
different types of stable equilibria is qualitatively similar if we
fix the value of $\phi$ and employ $b$ as the bifurcation parameter
instead. In fact, the symmetric dispersion can be shown analytically
to also undergo a pitchfork bifurcation for a critical value of $b$.} The six scenarios analysed in this Section are as follows: 
\begin{align*}
\textrm{(i).} & \left(\lambda,\gamma,\sigma,b\right)=\left(2,0.9,8,0.33\right); & \textrm{(ii).} & \left(\lambda,\gamma,\sigma,b\right)=\left(2,0.9,8,0.338\right);\\
\\\textrm{(iii).} & \left(\lambda,\gamma,\sigma,b\right)=\left(2,0.9,8,0.339\right); & \textrm{(iv).} & \left(\lambda,\gamma,\sigma,b\right)=\left(2,0.9,8,0.35\right);\\
\\\textrm{(v).} & \left(\lambda,\gamma,\sigma,b\right)=\left(2,0.9,8,0.55\right); & \textrm{(vi).} & \left(\lambda,\gamma,\sigma,b\right)=\left(4,0.9,8,0.55\right);
\end{align*}

\begin{figure}[!h]
\begin{centering}
\subfloat[Bifurcation diagram for scenario (i).]{\includegraphics[scale=0.8]{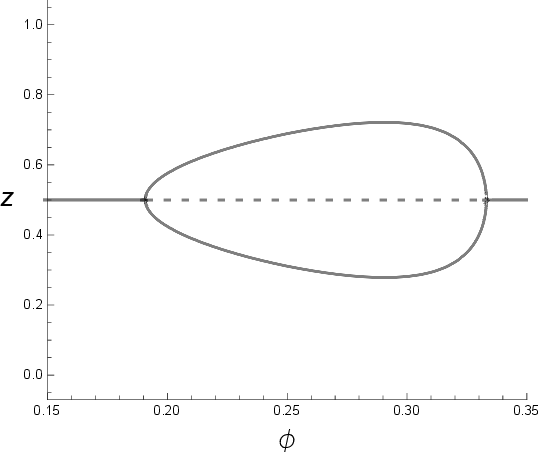}

\label{bifdiagramscenarioi}

}\subfloat[Bifurcation diagram for scenario (ii).]{\includegraphics[scale=0.8]{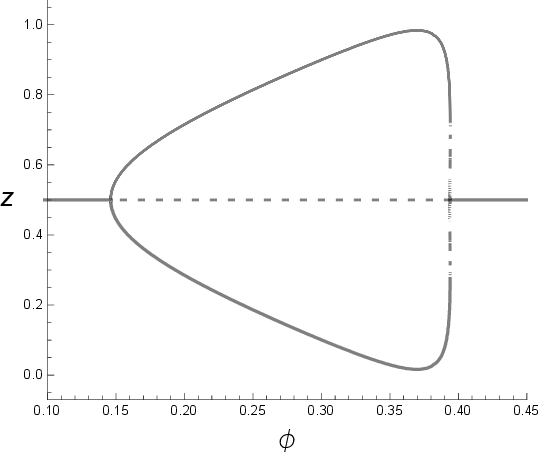}

\label{bifdiagramscenarioii} 

}
\par\end{centering}
\centering{}\subfloat[Bifurcation diagram for scenario (iii).]{\includegraphics[scale=0.8]{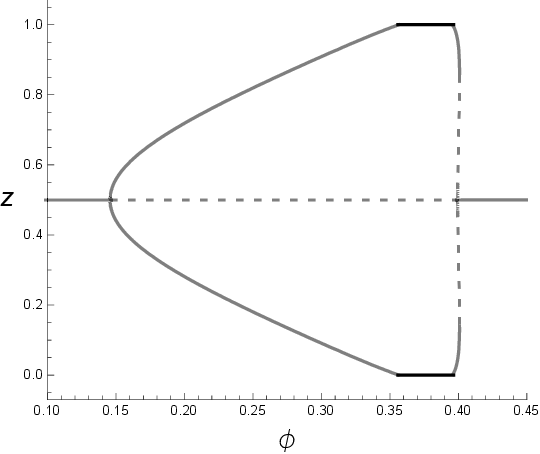}

\label{bifdiagramscenarioiii} 

}\subfloat[Bifurcation diagram for scenario (iv).]{\includegraphics[scale=0.8]{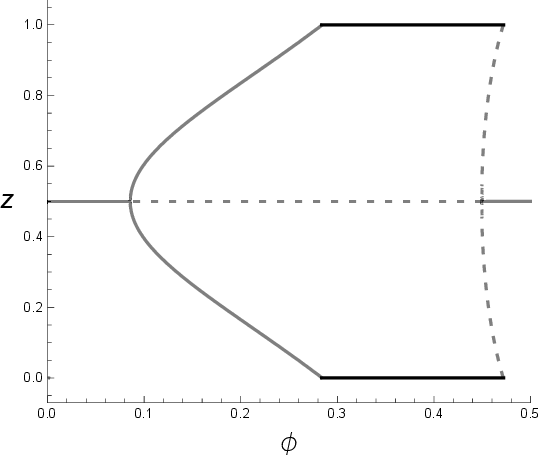}

\label{bifdiagramscenarioiv} 

}\caption{Bifurcation diagrams for $b<\tfrac{1}{2}:$ scenarios (i)-{}-(iv).
Filled lines correspond to stable equilbria and dashed lines correspond
to unstable equilibria.}
\end{figure}

\begin{figure}[!h]
\begin{centering}
\subfloat[Bifurcation diagram for scenario (v).]{\includegraphics[scale=0.8]{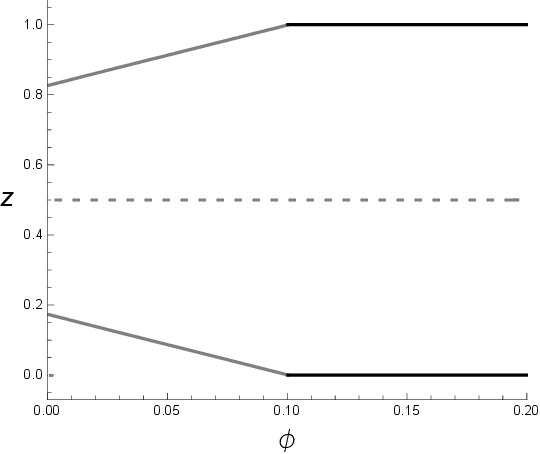}

\label{bifdiagramscenariov} 

}\subfloat[Bifurcation diagram for scenario (vi).]{\includegraphics[scale=0.8]{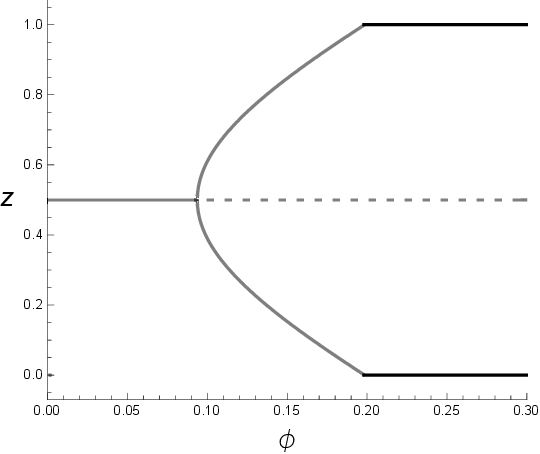}

\label{bifdiagramscenariobhigh} 

}
\par\end{centering}
\caption{Bifurcation diagrams for $b>\tfrac{1}{2}:$ scenarios (v)-{}-(vi).
Filled lines correspond to stable equilbria and dashed lines correspond
to unstable equilibria.}
\end{figure}

In scenario (i), shown in Figure~\ref{bifdiagramscenarioi}, the
intensity of local spillovers is such that within-region interaction
is relatively less important ($b=0.33$). For a low freeness of trade,
symmetric dispersion is stable because firms wish to avoid the burden
of a very costly transportation supplying to farmers from full agglomeration
in a single region. As $\phi$ increases, the economy initially agglomerates,
but then re-disperses as $\phi$ increases further. This re-dispersion
process occurs because, for a very high economic integration, firms
find it profitable to relocate to the less industrialized region in
order to benefit from the pool of researchers in the more agglomerated
region, which generates a higher innovation rate and thus higher wages
paid to researchers. Noteworthy, the turning point in the agglomeration
process happens \textit{before} industry reaches full agglomeration
in a single region, as in \citet{Pfluger-Suedekum-JUE2008}. However,
contrary to the latter, our model does not predict full agglomeration
in the entire parameter range of economic integration when the intensity
of local spillovers is low enough. Re-dispersion in scenario (i) is
more akin to geographical economic models of vertical linkages between
upstream and downstream firms by \citet{Krugman-Venables-QJE1995,venables1996equilibrium}
and \citet{puga1999rise}. However, in these models, re-dispersion
is smooth altogether and occurs when all agents are inter-regionally
\emph{immobile} and firms become too sensitive to regional cost differentials
when economic integration is very high.

In scenario (ii), illustrated by Figure~\ref{bifdiagramscenarioii},
local spillovers are just slightly stronger, and the model still accommodates
for re-dispersion. However, the re-dispersion process is not smooth
-- the economy suddenly jumps to symmetric dispersion from a fairly
asymmetric equilibrium spatial distribution.

Scenario (iii) also just slightly increases the weight of local spillovers
compared to the previous scenario (see Figure~\ref{bifdiagramscenarioiii}),
and the story of spatial outcomes as economic integration increases
is very similar, except that, in this case, full agglomeration is
stable for a small range of intermediate values of $\phi$, as predicted
by Proposition 3. The parametrization here also corresponds to that
illustrated in Figure~\ref{fig:Figure1}.

The re-dispersion processes of scenarios (ii) and (iii) are uncommon
in the literature of economic geography; rather, such jumps occur
in early models \citep{fujita1999spatial,BaldwinForslidMartinOttavianoRobertNicoud+2003}
from the state of symmetric dispersion to catastrophic agglomeration
\citep{behrens2011tempora}. The reverse discontinuous jump, i.e.,
from symmetric dispersion to partial agglomeration as trade costs
steadily decrease, has been uncovered in the model by \citet{pfluger2010size},
where \emph{all} production factors, except land, which is used both
for housing and production, are inter-regionally mobile. Their conclusions
about spatial outcomes reveal a line-symmetry of scenario (iii): as
integration increases, the economy jumps discontinuously from symmetric
dispersion to partial agglomeration, and the ensuing re-dispersion
is gradual and continuous.\footnote{In our model, the assumption that unskilled workers are immobile is
useful for tractability, as is the case of all footloose entrepreneur
models \citep{BaldwinForslidMartinOttavianoRobertNicoud+2003}. However,
we make the reasonable conjecture that immobile labour generates an
unnecessary dispersion force that changes the conclusions of our model
compared to the case of a perfectly mobile workforce only in the sense
of ``reversed'' stability as transport costs decrease, i.e. the
line-symmetry of all scenarios (i)--(vi). }

Figure~\ref{bifdiagramscenarioiv} illustrates scenario (iv) and
shows that the sudden re-dispersion process under a slightly higher
$b$ now happens from the state of full agglomeration directly to
the state of symmetric dispersion. In both scenarios (iii) and (iv),
the state of agglomeration is stable for intermediate values of economic
integration, as in \citet{robert2008offshoring}.

Noteworthy, throughout scenarios (i)--(iv), where global spillovers
dominate ($b<\tfrac{1}{2}$), the agglomeration process occurs gradually
for $\phi$ above the break point $\phi_{b1}$, as predicted by Proposition
6. At $\phi=\phi_{2b},$ the pitchfork bifurcation may be either supercritical
(e.g.\ Figure~\ref{bifdiagramscenarioi}) or subcritical (e.g.\ Figure~\ref{bifdiagramscenarioiii}).

Besides the bifurcation at symmetric dispersion, in Figures~\ref{bifdiagramscenarioii}
and \ref{bifdiagramscenarioiii} (scenarios (ii) and (iii)), a limit
point $\phi\equiv\phi_{l}\in(\phi_{b2},1)$ is discernible at which
two asymmetric equilibria, along a curve tangent to $\phi_{l}$ that
lies to its left, collide and coalesce. This suggests that in both
scenarios (ii) and (iii) the model undergoes a saddle-node bifurcation
at some asymmetric equilibrium $z^{*}\in\left(\frac{1}{2},1\right)$.
However, its existence seems impossible to determine analytically.
This kind of bifurcation also appears in the two-region footloose
entrepreneur model by \citet{Forslid-Ottaviano-JEG2003} with heterogeneous
agents analysed by \citet{Castro_2021} and also in the \citet{pfluger2004simple}
model extended to multiple regions by \citet{gaspar2018agglomeration}.
This bifurcation is associated with discontinuous jumps between some
asymmetric equilibrium other than agglomeration and the symmetric
dispersion once $\phi$ rises (falls) above (below) some threshold
level.

In scenario (v), for sufficiently strong local spillovers ($b>\tfrac{1}{2}$),
within-region interaction among researchers improves the innovation
rate enough such that wages become higher when researchers are either
partially agglomerated in one region for low values of $\phi$, or
completely agglomerated in one region for a high enough $\phi$. This
is portrayed in Figure~\ref{bifdiagramscenariov}. Scenario (v) precludes
the so-called ``no black-hole condition'' \citep{fujita1999spatial},
a condition that there always exists a region in parameter space for
which symmetric dispersion is stable. As argued by \citet{gaspar2018agglomeration},
this condition may be unwarranted if its exclusion allows for spatial
outcomes other than ubiquitous agglomeration.

We can thus conclude that a higher intensity of local spillovers is
associated with a more pronounced agglomeration during the industrialization
process, especially for intermediate values of economic integration.
If local spillovers dominate global spillovers ($b>\tfrac{1}{2}$),
re-dispersion is no longer possible because within-region interaction
among researchers is too important to make any deviation to the smaller
region worthwhile. This result is confirmed analytically in Section
9 for a more general specification of the spillover function $F_{i}(z,1-z)$
under mild assumptions.

In scenario (vi) we illustrate the qualitative change in the spatial
structure of the economy as $\phi$ increases for $b>1/2$, but with
a higher $\lambda$, since, with the parameter values of the previous
scenario, agglomeration would be ubiquitously stable (and hence uninteresting)
for higher values of $b$. In Figure~\ref{bifdiagramscenariobhigh},
we can observe a supercritical pitchfork bifurcation, as predicted
by Proposition 6 as in the model by \citet{pfluger2004simple}, where
there is no innovation. That is, for low levels of economic integration,
symmetric dispersion is stable. As $\phi$ increases, one region smoothly
becomes more and more industrialized en route to a full agglomeration
whereby that region becomes a core.

\begin{figure}[h]
\centering{}\includegraphics[scale=0.8]{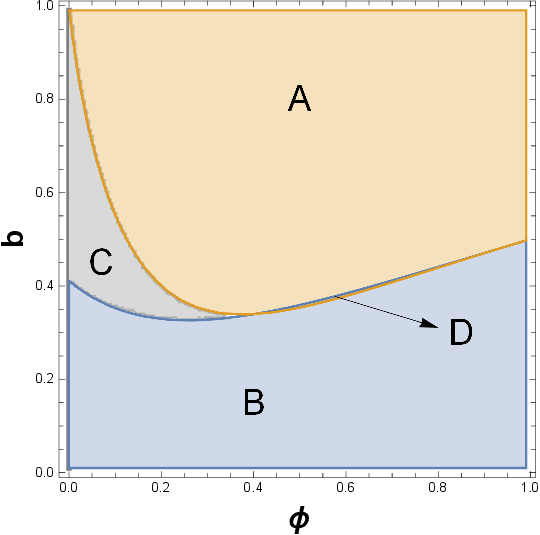} \includegraphics[scale=0.8]{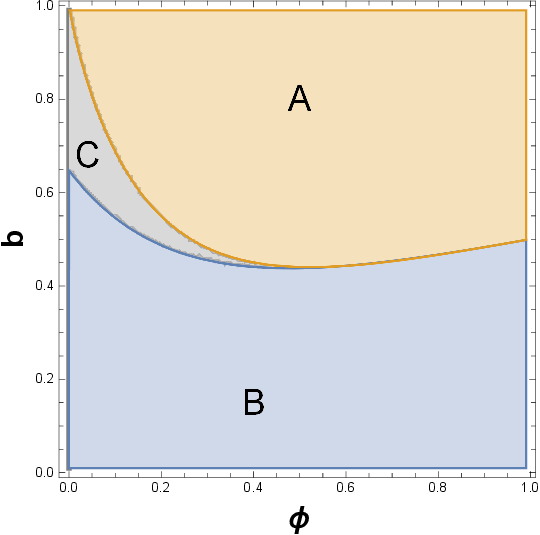}\caption{Stability region for each equilibrium in $(\phi,b)$-space. Region
$A$ denotes stable agglomeration, region $B$ denotes stable symmetric
dispersion, region $C$ denotes stable asymmetric equilbria. In region
$D$ both agglomeration and symmetric dispersion are stable. Parameter
values are $(\sigma,\gamma,\lambda)=(8,0.9,2)$, to the left, and
$(\sigma,\gamma,\lambda)=(8,0.9,4)$, to the right.}
\label{phibspace} 
\end{figure}

It is clear that the transition along and between each type of stable
equilibrium as $\phi$ increases depends on the intensity of local
spillovers $b$. We portray the stability regions of the three types
of equilibria (agglomeration, symmetric dispersion and asymmetric
dispersion) in Figure~\ref{phibspace}, in the space of $(\phi,b)\in(0,1)\times[0,1]$.
In region $A$, agglomeration is stable, which requires a high $\phi$
and a high $b$. In region $B$ we have stable symmetric dispersion,
which requires a low $b$ and either a very low or a very high value
of $\phi$ (as discussed before). Region $C$ denotes stable asymmetric
equilibria for low levels of $\phi$ and high enough levels of $b$.
There is a thin region $D$ where regions $A$ and $B$ overlap, indicating
that both agglomeration and symmetric dispersion are simultaneously
stable in that region. Finally, the region where both asymmetric dispersion
and symmetric dispersion are stable is not visually discernible, which
makes it a very unlikely outcome.

To conclude, we recall the results from Section 5.2 to analyze the
impacts of freer trade on knowledge variables. We saw that when $b>\tfrac{1}{2}$,
the growth rate of $a_{max}$ is maximized at agglomeration. Specifically,
we have 
\[
g=\frac{\omega b}{2\omega b+1},
\]
where $\omega=\frac{\ln\delta}{\delta\alpha}$. The higher the intensity
of local spillovers, the higher the growth rate $g$. The average
distance to the frontier $1-\bar{a}$ is also maximized when agglomeration
is stable. For $z=1$, regional aggregate knowledge levels are given
by $A_{2}=0$ and 
\[
A_{1}=\tfrac{1}{c}\bar{a}=\frac{a_{max}}{\alpha\left(1+2\omega b\right)}.
\]
For $b<\tfrac{1}{2},$ the growth rate $g$ is maximized at symmetric
dispersion: 
\[
g=\tfrac{1}{2}\left(\frac{\omega}{\omega+1}\right).
\]
The average distance to the frontier $1-\bar{a}$ is also maximized
when symmetric dispersion is stable. Regional aggregate knowledge
levels for $z=\tfrac{1}{2}$ are given by 
\[
A_{1}=A_{2}=\tfrac{1}{2c}\bar{a}=\frac{a_{\text{max}}}{2\alpha(\omega+1)}.
\]
The following result summarizes the trade-off between growth and average
relative quality. 
\begin{prop}
A high enough freeness of trade always yields a stable spatial distribution
that maximizes the global frontier growth rate $g$ and minimizes
the average relative quality $\bar{a}$. 
\end{prop}
Therefore, no matter the weight of local spillovers relative to global
spillovers, a high enough economic integration always maximizes the
growth rate of aggregate knowledge levels in the long-run and also
maximizes the time-invariant average distance to the frontier.



\section{Comparative statics}

It is worthwhile investigating analytically how changes in the intensity
of local spillovers affect the long-run spatial outcomes in the economy.
Following \citet{Castro_2021}, we say that a change in $b$ \emph{favours
agglomeration} if, due to the change: (a) symmetric dispersion may
become unstable but not stable, (b) agglomeration may become stable
but not unstable, and (c) asymmetric dispersion becomes more asymmetric.

Using (\ref{eq:indirectutilityregion1}) and (\ref{eq:utility differential}),
we have: 
\[
\frac{\partial\Delta v}{\partial b}=\frac{\gamma\mu(2z-1)(\phi+1)\mathinner{{\color{red}{\normalcolor \left\{ \lambda-4z^{2}+\phi\left[\lambda+4(z-1)z+2\right]+4z\right\} }}}}{2\sigma\left[z(\phi-1)+1\right]\left[z(1-\phi)+\phi\right]},
\]
which is positive for $z\in\left(\frac{1}{2},1\right)$. We have the
following result. 
\begin{lem}
An increase in the intensity of local spillovers $b$ favours agglomeration. 
\end{lem}
\begin{proof}
See Proposition 9 of \citet[p.197]{Castro_2021}.  
\end{proof}
The interpretation behind this result is straightforward: a higher
$b$ implies a higher innovation rate in region $i$ when more researchers
live in region $i$. Hence,{} the expected payoff from innovation
is higher, which leads to higher wages and, thus, a higher utility
differential in region $i$. Therefore, a higher $b$ makes stability
of agglomeration (symmetric dispersion) more (less) likely for a given
value of $\phi$, and asymmetric dispersion becomes more asymmetric.

Suppose $b$ is such that symmetric dispersion is the unique stable
equilibrium for a low $\phi$ (Proposition 4), asymmetric dispersion
is the unique stable equilibrium for intermediate values of $\phi$
(Proposition 5), and agglomeration is the unique stable equilibrium
for a high $\phi$ (Proposition 3). Accordingly, the model undergoes
a supercritical pitchfork bifurcation at symmetric dispersion (Proposition
6). As illustrated in Figure~\ref{supercriticalpitchfork}, an increase
in $b$ shifts the pitchfork bifurcation leftwards, provided that
the set of stable equilibria remains unchanged. Conversely, a decrease
in $b$ promotes dispersion, which fits with the intuition that the
increasing globalization of knowledge may promote convergence across
regions \citep{carlino2015agglomeration}.


\begin{figure}[h]
\centering{}\includegraphics{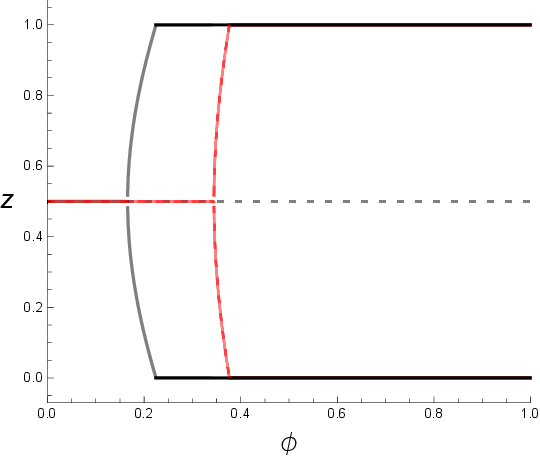}\caption{Bifurcation diagrams: (i) to the left, in black, we set $b=0.9$;
(ii) to the right, in red, we set $b=0.7$. Parameter values are $(\sigma,\gamma,\lambda)=(8,1,4).$}
\label{supercriticalpitchfork} 
\end{figure}

\section{Spillovers and re-dispersion: a more general case}

Let us now consider the more general specification for the innovation
rate in (\ref{eq:phiFi}). In region $1$, we have 
\begin{equation}
\Phi(z)=\frac{\bar{a}}{\delta}F(z;\boldsymbol{b}),\label{eq:general}
\end{equation}
where $\boldsymbol{b}$ is a vector of parameters that \emph{does}
\emph{not }include the freeness of trade $\phi$, but includes the
relative intensity of local knowledge spillovers $b\in(0,1)$. Additionally,
assume that there exists $\bar{b}\in(0,1)$ such that $F^{\prime}(z)<0$
if $b<\bar{b}$, and $F^{\prime}(z)>0$ if $b>\bar{b}$, for $z\in[\frac{1}{2},1]$.
This reflects the fact that $F(z)$ constitutes a dispersion force
in the larger region when global spillovers dominate ($b<\bar{b}$),
and an agglomeration force in the larger region when local spillovers
are stronger ($b>\bar{b}$). Again, a lower $b$ may be associated
with congestive effects in the production of knowledge in the larger
region \citep{brezis1993}, which creates opportunities for breakthrough
innovations in smaller regions. The indirect utility in region $1$
is obtained by replacing the term $\gamma\left[bz+(1-b)(1-z)\right]$
in (\ref{eq:indirectutilityregion1}) with $F(z)$.

At symmetric dispersion, making use of the fact that{} $F_{1}^{\prime}(\tfrac{1}{2})=-F_{2}^{\prime}(\frac{1}{2}),$
we have at most two break points: 
\begin{align*}
\phi_{b1}^{G} & =\frac{(\lambda+1)(\sigma-1)\left[F'\left(\frac{1}{2}\right)+2F\left(\frac{1}{2}\right)\right]-\varphi}{2\left[F\left(\frac{1}{2}\right)(\lambda+2)(\sigma-1)+\sigma\right]-(\lambda+1)(\sigma-1)F'\left(\frac{1}{2}\right)},\\
\phi_{b2}^{G} & =\frac{(\lambda+1)(\sigma-1)\left[F'\left(\frac{1}{2}\right)+2F\left(\frac{1}{2}\right)\right]+\varphi}{2\left[F\left(\frac{1}{2}\right)(\lambda+2)(\sigma-1)+\sigma\right]-(\lambda+1)(\sigma-1)F'\left(\frac{1}{2}\right)},
\end{align*}
where 
\[
\varphi=2\sqrt{2F\left(\frac{1}{2}\right)(\lambda+1)^{2}(\sigma-1)^{2}F'\left(\frac{1}{2}\right)+\left[F\left(\frac{1}{2}\right)(\sigma-1)+\sigma\right]^{2}}.
\]
Since $F(\frac{1}{2})>0,$ it can be shown that $\phi_{b2}^{G}\notin(0,1)$
if $F^{\prime}(\frac{1}{2})>0,$ i.e., if $b>\bar{b}$. Since symmetric
dispersion is stable for $\phi\in(0,\phi_{b1}^{G})\cup(\phi_{b2}^{G},1),$
we can thus conclude that the complete re-dispersion of economic activities
is not possible if local spillovers dominate. This leads to the following
result. 
\begin{prop}
If the relative intensity of local spillovers is high at symmetric
dispersion, then a (complete) re-dispersion of economic activities
for a high economic integration is not possible. 
\end{prop}
In other words, the process of (complete) re-dispersion depends on
the dispersive or agglomerative nature of knowledge spillovers at
symmetric dispersion. When $b>\bar{b}$, knowledge spillovers are
localized and constitute an agglomeration force in the larger region,
which becomes relatively stronger as trade barriers decrease. For
a sufficiently high $\phi$, re-dispersion is impossible because there
is no incentive to relocate to the smaller region where innovation
is less likely and expected profits are lower.

The result from Proposition 8 is consistent with that of Proposition
4 where $\bar{b}=\tfrac{1}{2}$. If local spillovers dominate, symmetric
dispersion can only be stable for low values of the freeness of trade,
i.e. for $\phi\in(0,\phi_{b1}^{G})$.

We now take one step further by investigating what type of bifurcation
the symmetric dispersion undergoes at some break point $\phi_{b}\in\left\{ \phi_{b1}^{G},\phi_{b2}^{G}\right\} $.
We have the following derivatives: 
\[
\dfrac{\partial f}{\partial z}\left(\dfrac{1}{2};\phi_{b}\right)=0;\ \dfrac{\partial^{2}f}{\partial z^{2}}\left(\dfrac{1}{2};\phi_{b}\right)=0;\ \dfrac{\partial f}{\partial\phi}\left(\dfrac{1}{2};\phi_{b}\right)=0;
\]
Furthermore, we have: 
\[
\dfrac{\partial^{2}f}{\partial\phi\partial z}\left(\dfrac{1}{2};\phi_{b}\right)=\frac{8\mu}{(\phi_{b}+1)^{3}}\left\{ \frac{\gamma F\left(\frac{1}{2}\right)\left[-2\lambda(\phi_{b}-1)-3\phi_{b}+1\right]}{\sigma}+\frac{\phi_{b}+1}{1-\sigma}\right\} 
\]
which is zero if and only if $F(\frac{1}{2})=F_{b},$ with 
\[
F_{b}\equiv-\frac{\sigma(\phi_{b}+1)}{\gamma(\sigma-1)\left[2\lambda(\phi_{b}-1)+3\phi_{b}-1\right]}.
\]
If $F(\frac{1}{2})\neq F_{b},$ then symmetric dispersion undergoes
a pitchfork bifurcation (see Appendix B.6). In this case, its criticality
is determined by the sign of the third-order derivative $\frac{\partial^{3}f}{\partial z^{3}}(\frac{1}{2};\phi_{b}).$
However, without additional details on the shape of $F(z)$ it is
impossible to convey further information on these conditions and to
relate them with $b$.

We conclude that, while the shift between distinct spatial outcomes
as economic integration increases may display differences across different
values of $b$, the presence of a (complete) re-dispersion phase should
depend solely on the dispersive or agglomerative nature of knowledge
spillovers, i.e.\ on whether spillovers are mainly global or local.

\section{On freer trade and global spillovers}

So far, we have assumed that inter-regional spillovers do not depend
on the trade freeness $\phi$. This accounts for the possibility of
re-dispersion above a certain level of economic integration because
inter-regional spillovers do not depend directly on $\phi$, but become
relatively stronger compared to agglomeration forces due to market
factors as $\phi$ increases beyond a certain level. Since dispersion
forces operating through trade linkages depend on $\phi$, we have
two types of dispersion forces: one that depends on $\phi$ (market
factors) and one that does not (regional spillovers). This fits the
narrative of \citet{akamatsu2025spatialscaleagglomerationdispersion}
regarding models with different types of dispersion forces.

However, one may argue that knowledge spillovers between different
regions are stronger when there is a higher level of economic integration
because it mitigates the distance decay effect, facilitating cross-regional
communication, collaboration and knowledge diffusion between researchers
and firms through trade \citep[see e.g.][]{keller2010international,bueraoberfield2020}.
Accordingly, suppose now that the weight of inter-regional spillovers
depends on the level of the freeness of trade. Specifically, let the
spillover function be $F_{i}=bz_{i}+\phi\psi(1-b)z_{j}$, so that
the innovation rate for a variety $s$ in region $i$ is given by
\begin{equation}
\Phi_{i}^{g}=\frac{\gamma}{\delta}\left[bz_{i}+\phi\psi(1-b)z_{j}\right]\bar{a},\label{eq:withintegration}
\end{equation}
where $\psi>0$ measures the impact of integration in the intensity
of inter-regional knowledge spillovers. We set $\psi=\gamma=1$ for
simplicity and without loss of generality. Following the same reasoning
as that of the proof of Proposition 1 (see Appendix B.1), it is possible
to show that, through inspection of the utility differential $\Delta v^{g}\equiv v_{1}^{g}(z)-v_{2}^{g}(z)$,
there are at most two long-run equilibria for $z>\frac{1}{2}.$ Further,
we can show analytically that dispersion is always unstable when $b>\tfrac{1}{2}$.
We also find that when global spillovers dominate, i.e.\ $\phi>\frac{b}{(1-b)}$,
then symmetric dispersion is stable (unstable) if $\phi>(<)\phi_{b}^{g}$,\footnote{The break-point $\phi_{b}^{g}$ is the solution to $\dfrac{\partial f}{\partial z}\left(\dfrac{1}{2};\phi_{b}^{g}\right)=0$.}
where 
\[
\phi_{b}^{g}\equiv\frac{2(\lambda+1)(\sigma-1)-\sqrt{4\sigma^{2}+4(\sigma-1)^{2}+8\sigma(\sigma-1)}}{(2\lambda+4)(\sigma-1)+2\sigma}.
\]
We assume that $\lambda>\frac{\sigma}{\sigma-1},$ so that $\phi_{b}^{g}\in\left(\frac{b}{(1-b)},1\right)$.
This is analogous to the ``no black-hole condition'' found in \citet{fujita1999spatial}.

Agglomeration is stable (unstable) when $\phi<(>)\phi_{s}^{g}$,\footnote{The sustain point $\phi_{s}^{g}$ is obtained from $\Delta v(1;\phi_{s}^{g})=0$.}
where $\phi_{s}^{g}$ satisfies 
\[
\Omega^{g}\equiv\frac{\mu\left[(b-1)(\lambda+2)\phi_{s}^{g}{}^{2}+3b\lambda+2b-\lambda\right]}{2\sigma}-\frac{\mu\log(\phi_{s}^{g})}{\sigma-1}=0.
\]
The sustain point $\phi_{s}^{g}\in\left(\frac{b}{(1-b)},1\right)$
is unique. We conclude that, as trade barriers decrease from a high
level, the shift in the spatial distribution now occurs from agglomeration
towards symmetric dispersion. The intuition is quite straightforward.
When economic integration is very low, spillovers are very localized,
implying agglomeration of researchers and firms in one region. A higher
freeness of trade enhances the intensity of knowledge spillovers between
different regions, which means that firms have incentive to relocate
to the smaller region from where they can leverage the concentrated
knowledge base of the larger region and avoid congestion effects on
innovation. 
\begin{figure}[h]
\centering{}\includegraphics[scale=0.6]{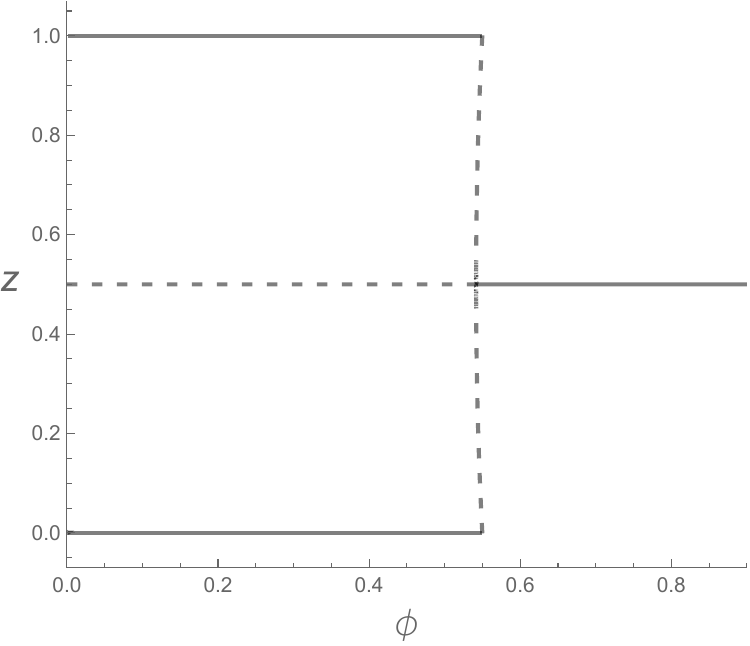}
\includegraphics[scale=0.6]{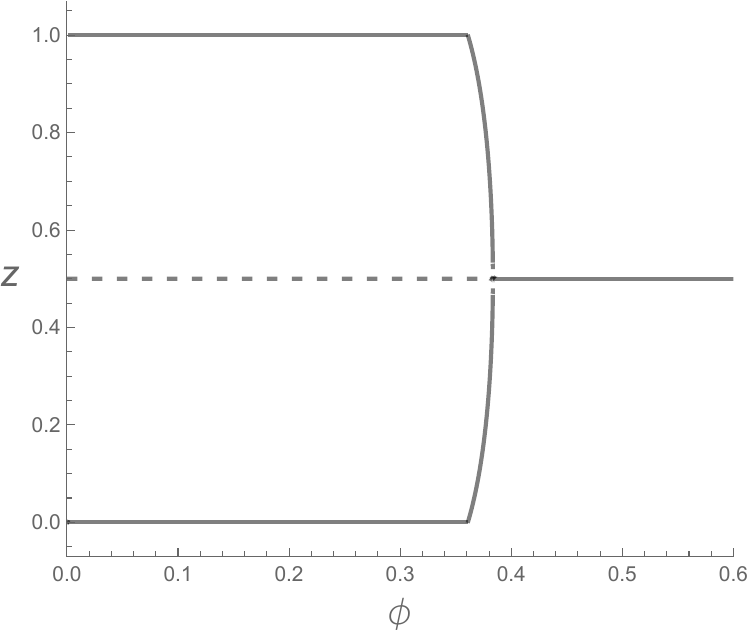}\caption{Bifurcation diagrams for $b<\tfrac{1}{2}$: (i) to the left, we set
$\lambda=2$; (ii) to the right, we set $\lambda=4$. Parameter values
are $(b,\sigma,\gamma,\psi)=(0.2,8,1,1).$}
\label{freertradepitchforks} 
\end{figure}

For a given $b$, the smoothness of the transition depends on the
fraction of agents that are mobile, captured by the term $\frac{1}{\lambda}$.
Figure \ref{freertradepitchforks} shows that the dispersion process
is sudden once $\phi>\phi_{s}^{g}$ if $\lambda=2$. If $\lambda=4$,
the dispersion process is smooth.\footnote{It can be shown analytically that the model undergoes a pitchfork
bifurcation at $\phi=\phi_{b}^{g}$. The criticality of the bifurcation
can be pinned down in terms of $b$. Above a certain threshold, the
pitchfork goes from supercritical to subcritical. Therefore, a higher
$b$ increases the likelihood of catastrophic dispersion for a high
enough $\phi$.} We conclude the following.
\begin{prop}
When global spillovers are magnified by freer trade (and $b<\tfrac{1}{2}$),
a progressive increase in economic integration leads to a monotone
dispersion process from agglomeration to symmetric dispersion.
\end{prop}
Our analysis is again in line with \citet{akamatsu2025spatialscaleagglomerationdispersion}
because the model now only contains dispersion forces that depend
on $\phi$ (market factors and regional spillovers).

Finally, using $F_{1}=F(z)=\gamma\left[bz+\psi\phi(1-b)(1-z)\right]$
and $F_{2}=F(1-z)$ in (\ref{eq:relativequalityss}) and (\ref{eq:frontiergrowthrate}),
it is possible to show that the growth rate of $a_{max}$ is maximized
at symmetric dispersion. Therefore, a freeness of trade above the
break-point yields a stable spatial distribution that maximizes $g$.
Specifically, we have 
\[
g=\frac{\omega\left[b(1-\phi)+\phi\right]}{2(\omega+2)},
\]
with $\omega=\frac{\ln\delta}{\delta\alpha}$. The average distance
to the frontier $1-\bar{a}$ is also maximized when symmetric dispersion
is stable. The regional aggregate knowledge levels are given by 
\[
A_{1}=A_{2}=\tfrac{1}{2c}\bar{a}=\frac{a_{max}}{\alpha\left(1+\omega\right)}.
\]
These results corroborate Proposition 7 and further illustrate how
integration increases the growth rate of the global frontier quality
$g.$

\section{Concluding remarks}

We have analyzed a two-region Schumpeterian spatial quality-ladder
model. We looked at how the spatial creation and diffusion of knowledge
and increasing returns in manufacturing interact to shape the spatial
economy. Knowledge levels translate into higher arrival rates of innovation
that continuously improve the quality of manufactures varieties and
push the global frontier quality level forward. In turn, innovation
rates depend on the spatial distribution of knowledge levels in the
economy. As the regional average quality levels relative to the global
frontier converge in the long-run, innovation depends crucially on
the spatially weighted spatial distribution of researchers across
regions. We thus introduce a spatial mechanism of knowledge spillovers
conveying the idea that public knowledge transfers imperfectly across
space, as argued by e.g. \citet{krugman1991geography} and supported
by the empirical evidence of \citet{audretsch1996R&D}.

We find that if global knowledge spillovers dominate, the model accounts
for (a complete) re-dispersion of economic activities after an initial
stage of progressive agglomeration as trade integration increases
from a low level. A higher weight of global spillovers captures the
idea that congestion dampens the production of knowledge, turning
innovation into a dispersion force, since researchers have incentives
to relocate to the smaller region and benefit from the sizable pool
of agents living in the larger region. Smaller regions tend to have
tightly--knit innovation ecosystems where researchers are geographically
closer. This proximity facilitates frequent interactions, speeding
up knowledge diffusion \citep{Doring01052006}. All this promotes
a more balanced distribution of researchers across regions for a given
level of economic integration. However, the relationship between economic
integration and spatial imbalances is far from trivial, as we have
shown a myriad of different qualitative possibilities regarding transitions
between different stable states that depend on the particular level
of the intensity of local spillovers relative to global spillovers.
This level determines the spatial distribution for a given level of
integration and impacts the smoothness of the transition from either
partial or full agglomeration to symmetric dispersion as trade integration
reaches high enough levels.

On the other hand, if local spillovers dominate, knowledge spillovers
are more localized and generate an additional agglomeration force
in the model. In this case, a higher economic integration leads to
progressive agglomeration, as in \citet{martin1999growing} and \citet{martin2001growth}.
Re-dispersion for a high freeness of trade is precluded because there
is no incentive to relocate to the smaller region where innovation
is less likely and expected profits are lower. Further, we show that
a higher intensity of local spillovers always favours more agglomerated
outcomes.

The empirical literature shows that knowledge tends to cluster in
space (see e.g. \citet{carlino2015agglomeration,kerr2020tech} and
references therein). However, increases in economic integration brought
by the ongoing globalization may increase the weight of inter-regional
interaction and thus lead to greater spatial convergence in the future.
In this sense, we have shown in Section 10 that a freer trade leads
to a monotone dispersion as global spillovers become relatively stronger.

The inclusion of a reduced form diffusion term in the dynamics for
innovation, leading to the convergence of regions' average distance
to the knowledge frontier, is admittedly a simplifying mechanism in
our model. However, it is warranted because it allows differences
in region to stem uniquely from the spatial distribution of researchers
and its impact on regional knowledge spillovers. This allows us to
focus on spatial outcomes as a result of self-reinforcing agglomeration
mechanisms stemming from the interaction between endogenous market
sizes and the economic geography of knowledge spillovers (as proposed
by \citet{BS2022}). However, we acknowledge that allowing for asymmetries
in regional average qualities relative to the global frontier to arise
endogenously is a promising avenue for future research.

Other topics for future research include the analysis of a multi-regional
setup that allows to take the model to the data and test its theoretical
predictions. In this sense, our paper may contribute to the literature
of Quantitative Spatial Economics by providing a new theoretical mechanism
that could be incorporated into larger-scale quantitative models to
better understand the determinants of the spatial distribution of
economic activity.

Finally, we argue that introducing forward-looking behaviour and inter-temporal
optimization would allow to reach deeper insights regarding regional
growth and its relation with spatial outcomes.


\section*{Acknowledgments}

We are thankful to Steven Bond-Smith, Sofia B. S. D. Castro, João
Correia da Silva, Pascal Mossay, Pietro Peretto, Perdo Reis, Anna
Rubinchik and Jorge Saraiva for very useful comments and suggestions.
We would also like to thank participants at several conferences for
their valuable feedback. Part of this study was developed while José
M. Gaspar was a researcher at the Research Centre in Management and
Economics, Católica Porto Business School, Universidade Católica Portuguesa.
Funding Information: Japan Society for the Promotion of Science Grant/Award
Number: 21K04299; Fundação para a Ciência e Tecnologia UID/04105/2023,
UIDB/00731/2020 and PTDC/EGE-ECO/30080/2017.

\appendix

\section*{Appendix A}

\subsection*{Proof of Theorem 1}

A steady-state requires $\frac{d\bar{a}_{1}}{dt}=\frac{d\bar{a}_{2}}{dt}=0,$
which, using (\ref{eq:relativequalitymotion}), implies 
\begin{align*}
\frac{1}{2}(\Phi_{1}+\Phi_{2})(1-\bar{a}_{1}) & =g\bar{a}_{1}.\\
\frac{1}{2}(\Phi_{1}+\Phi_{2})(1-\bar{a}_{2}) & =g\bar{a}_{2}.
\end{align*}
Hence, the symmetric state $(\bar{a}_{1},\bar{a}_{2})=(\bar{a},\bar{a})\in\mathcal{H}$
is a steady-state of (\ref{eq:relativequalitymotionfinal}) and is
unique. \medskip{}

\noindent\textbf{\emph{Global stability}}\emph{.} Let $x=\frac{\bar{a}_{1}}{\bar{a}_{2}}\in\mathbb{R}_{+}.$
The dynamics for the ratio $x$ are given by: 
\[
\dot{x}=\frac{1}{\bar{a}_{2}}\left(\frac{d\bar{a}_{1}}{dt}-x\frac{d\bar{a}_{2}}{dt}\right).
\]
Using (\ref{eq:relativequalitymotionfinal}) and setting $\bar{a}_{1}=x\bar{a}_{2}$,
we get 
\[
\dot{x}=\frac{\Phi_{1}+\Phi_{2}}{2\bar{a}_{2}}(1-x).
\]
Define the Lyapunov function $V(x)=\frac{1}{2}(x-1)^{2}.$ Then 
\[
\dot{V}=\dot{x}(1-x)=-\frac{\Phi_{1}+\Phi_{2}}{2\bar{a}_{2}}(x-1)^{2},
\]
which is strictly negative for $x\neq1.$ Hence, we have $x(t)\rightarrow1$
from any interior initial condition, i.e., $A_{1}/A_{2}\rightarrow1$
and thus $x=1$ is globally stable for the ratio dynamics.\medskip{}

\noindent\textbf{\emph{Average relative quality}}\emph{.} An interior
steady-state with $\bar{a}_{i}=\bar{a}\in(0,1)$ solves 
\[
\frac{1}{2}\left(\Phi_{1}+\Phi_{2}\right)(1-\bar{a})-g\bar{a}=0,
\]
which implies 
\[
\bar{a}=\frac{\Phi_{1}+\Phi_{2}}{\Phi_{1}+\Phi_{2}+2g},
\]
where $g=\ln\delta\left(n_{1}\Phi_{1}+n_{2}\Phi_{2}\right)$. Since
$F_{i}$ is homogeneous of degree $1,$ we have $F_{i}(A_{i},A_{j})=\frac{a_{max}}{\alpha}\bar{a}F_{i}(z_{i},z_{j})$.
Moreover, $G_{1}(1)=G_{2}(1)=1$. Let $F_{i}\equiv F_{i}(\Vtz)$.
Then the frontier growth $g$ is given by 
\[
g=\frac{\ln\delta}{\alpha^{2}}a_{max}\bar{a}\xi(c)\left(z_{1}F_{1}+z_{2}F_{2}\right).
\]
With $\xi\left(c\right)=\frac{1}{\delta}\frac{\alpha}{a_{max}}$,
this yields 
\begin{align*}
\bar{a} & =\frac{F_{1}+F_{2}}{F_{1}+F_{2}+\frac{2}{\alpha\delta}\ln\delta\left[zF_{1}+(1-z)F_{2}\right]},\\
g & =\frac{\ln\delta}{\alpha\delta}\bar{a}\left(z_{1}F_{1}+z_{2}F_{2}\right).
\end{align*}
completing the proof.\hfill{}$\square$

\section*{Appendix B}

This appendix contains the more cumbersome formal proofs that support
our main results from Section 6 onward.

\subsection*{B.1 Proof of proposition 1}

Differentiating $\Delta v(z)$ in (\ref{eq:utility differential})
yields: 
\begin{equation}
\frac{d\Delta v}{dz}(z)=\frac{\mu P(z)}{2(\sigma-1)\sigma\left[z(\phi-1)+1\right]{}^{2}\left[z(1-\phi)+\phi\right]{}^{2}},\label{prop 1 eq}
\end{equation}
where: 
\[
P(z)=a_{1}z^{4}+a_{2}bz^{3}-2(1-\phi)a_{3}z^{2}+2(1-\phi)a_{4}z+a_{5},
\]
with:{\small{} 
\begin{align*}
a_{1}= & 4(1-2b)\gamma\mu(\sigma-1)(\phi-1)^{3}(\phi+1)\\
a_{2}= & 8(2b-1)\gamma\mu(\sigma-1)(\phi-1)^{3}(\phi+1)\\
a_{3}= & \gamma(\sigma-1)\left\{ b(\phi+1)\left[(\lambda-2)\phi^{2}-\lambda+18\phi-4\right]-\phi\left[\lambda(\phi-1)\phi+\lambda+6\phi\right]+\lambda-8\phi+2\right\} \\
 & +\sigma(\phi+1)(\phi-1)^{2}\\
a_{4}= & \gamma(\sigma-1)\left\{ \lambda(\phi-1)\left[b(\phi+1)^{2}-\phi^{2}-1\right]+2\phi\left[b(\phi+1)(\phi+5)-\phi(\phi+2)-3\right]\right\} \\
 & +\sigma(\phi+1)(\phi-1)^{2}\\
a_{5}= & \gamma(\sigma-1)\left\{ \lambda\left(\phi^{2}+1\right)\left[b(\phi+1)^{2}-\phi^{2}-1\right]+2\phi\left[b\left(\phi^{3}+3\phi^{2}+\phi-1\right)-\phi\left(\phi^{2}+\phi+1\right)+1\right]\right\} \\
 & -2\sigma\phi\left(\phi^{2}-1\right).
\end{align*}
}{\small\par}

\noindent The denominator of (\ref{prop 1 eq}) is positive, which
means that the sign of $\frac{d\Delta v}{dz}(z)$ is given by the
sign of $P(z)$, which is a fourth degree polynomial in $z$. Therefore,
$\Delta v(z)$ has at most four turning points, and thus at most five
equilibria for $z\in[0,1]$. We know that $z=\frac{1}{2}$ is an invariant
pattern. By symmetry, we can establish that there exist at most two
equilibria for $z>\tfrac{1}{2}$, which concludes the proof.\hfill{}$\square$

\subsection*{B.2 Proof of Proposition 2}

Proceeding as in \citet{gaspar2018agglomeration,Gaspar-et-al-RSUE2021},
the equilibrium condition $\Delta v(z)=0$ yields: 
\begin{equation}
\lambda\equiv\lambda^{*}(z)=-2\frac{c_{1}c_{2}+c_{3}\ln\left[\frac{z(\phi-1)+1}{z(1-\phi)+\phi)}\right]}{c_{4}},\label{eq:equilibriumcondition}
\end{equation}
where: 
\begin{align*}
c_{1} & =\gamma(\sigma-1)(2z-1)\\
c_{2} & =\phi^{2}\left[2b(z-1)z+b-z^{2}+z-1\right]+(1-2b)(z-1)z+b\phi\\
c_{3} & =\sigma\left[z(\phi-1)+1\right]\left[z(\phi-1)-\phi\right]\\
c_{4} & =\gamma(\sigma-1)(2z-1)\left[b(\phi+1)^{2}-\phi^{2}-1\right].
\end{align*}
It is easy to note that $\lambda^{*}(z)$ has a vertical asymptote
if and only if $c_{4}=0$, i.e., iff: 
\[
b=\hat{b}\equiv\frac{\phi^{2}+1}{(\phi+1)^{2}}.
\]

\noindent For $z\in\left(\frac{1}{2},1\right],$ the log term of $\lambda^{*}(z)$
is negative, as is $c_{3}$. Next, we have $c_{1}>0$, and $c_{2}>0$
if: 
\[
b\geq\underline{b}\equiv\frac{(z-1)z\left(\phi^{2}-1\right)+\phi^{2}}{2(z-1)z\left(\phi^{2}-1\right)+\phi(\phi+1)},
\]
where $0<\underline{b}<\hat{b}$. Since $c_{4}<0$ only if $b<\hat{b}$,
we have $\lambda^{*}(z)>0$ if $b\in\left[\underline{b},\hat{b}\right)$
and $\lambda^{*}(z)<0$ if $b\in\left(\hat{b},1\right)$. For $b\in\left(0,\underline{b}\right)$,
we need further inspection.

We have that:{\footnotesize{} 
\[
\frac{\partial\lambda^{*}}{\partial b}(z)=\frac{2(\phi+1)\left[z(\phi-1)+1\right]\left[z(\phi-1)-\phi\right]\left\{ \gamma(\sigma-1)(2z-1)(\phi-1)+\sigma(\phi+1)\ln\left[\frac{z(\phi-1)+1}{z(1-\phi)+\phi)}\right]\right\} }{\gamma(\sigma-1)(2z-1)\left[-b(\phi+1)^{2}+\phi^{2}+1\right]^{2}},
\]
}which is positive for all $z\in\left(\frac{1}{2},1\right].$ The
unique zero of $\lambda^{*}(z)$ in terms of $b$ is given by:{\footnotesize{}
\[
b=\tilde{b}\equiv\frac{\gamma(\sigma-1)(2z-1)\left[(z-1)z\left(\phi^{2}-1\right)+\phi^{2}\right]-\sigma\left[z(\phi-1)+1\right]\left[z(\phi-1)-\phi\right]\ln\left[\frac{z(\phi-1)+1}{z(1-\phi)+\phi)}\right]}{\gamma(\sigma-1)(2z-1)(\phi+1)\left[2(z-1)z(\phi-1)+\phi\right]},
\]
} with $\tilde{b}<\underline{b}$ and $\lambda^{*}(z)>0$ for $b\in(\tilde{b},\hat{b})$.
It is possible to show that $\tilde{b}$ is increasing in $\gamma.$
Moreover, we have $\tilde{b}=0$ if and only if: 
\[
\gamma=\gamma_{c}\equiv\frac{\sigma\left[z(1-\phi)-1\right]\left[z(1-\phi)+\phi\right]\ln\left[\frac{z(\phi-1)+1}{z(1-\phi)+\phi}\right]}{(\sigma-1)(2z-1)\left[(z-1)z\left(\phi^{2}-1\right)+\phi^{2}\right]}>0.
\]
This means that $\tilde{b}\geq0$ if $\gamma\geq\gamma_{c}$ and $\tilde{b}<0$
if $\gamma<\gamma_{c}$ and $\gamma_{c}\in(0,1]$. Since $\gamma_{c}\in(0,+\infty)$,
we have $\tilde{b}<0$ if $\gamma_{c}>1$. As a result, we have $\tilde{b}<0$
if $\gamma\in\left(0,\min\{1,\gamma_{c}\}\right)$ and $\tilde{b}\geq0$
if $\gamma\in\left[\min\{1,\gamma_{c}\},1\right)$. Then $\lambda^{*}(z)>0$
if $\gamma\in\left(0,\min\{1,\gamma_{c}\}\right)$ and $b\in(0,\hat{b})$.
Otherwise, we have $\lambda^{*}(z)>0$ if $\gamma\in\left[\min\{1,\gamma_{c}\},1\right)$
and $b\in(\tilde{b},\hat{b}).$ Therefore, $\lambda^{*}(z)$ is positive
for $b\in\left(\max\left\{ 0,\tilde{b}\right\} ,\hat{b}\right)$ and
negative for $b\in\ \left(0,\max\left\{ 0,\tilde{b}\right\} \right)\cup\left(\hat{b},1\right)$,
where $\max\left\{ 0,\tilde{b}\right\} $ depends on $\gamma_{c}$
and on the value of $\gamma$ as described above.

Thus, we can assert that, if $b\in\left(\max\left\{ 0,\tilde{b}\right\} ,\hat{b}\right)$,
there exists a value of $\lambda>0$ such that at least one (at most
two) dispersion equilibrium $z\equiv z^{*}\in\left(\frac{1}{2},1\right]$
exists. This concludes the proof.\hfill{}$\square$

\subsection*{B.3 Proof of Proposition 3}

We have: 
\[
\lim_{\phi\rightarrow0^{+}}\Omega(\phi)=-\infty\text{ and }\text{\ensuremath{\Omega(1)=\frac{\gamma(2b-1)(\lambda+1)}{\sigma}.}}
\]
Therefore, $S(1)>0$ if $b>\tfrac{1}{2}$ and we conclude that $\mathcal{S}(\phi)$
has at least one zero for $\phi\in\left(0,1\right)$. Further, we
have: 
\[
\dfrac{d\Omega}{d\phi}(\phi)=\dfrac{1}{2\phi^{2}}\left\{ \frac{\gamma(b-1)\left[\lambda\left(\phi^{2}-1\right)+2\phi^{2}\right]}{\sigma}-\frac{2\phi}{\sigma-1}\right\} ,
\]
whose sign depends on that of the second term, which is a second degree
polynomial and thus has at most two zeros $\left\{ \phi^{-},\phi^{+}\right\} $,
with $\phi^{+}>\phi^{-}$. However, only $\phi^{+}$ lies on the interval
$\phi\in\left(0,1\right)$: 
\[
\phi^{+}=\frac{\sigma\left[\frac{1}{\sigma-1}-\sqrt{\frac{\gamma^{2}(b-1)^{2}\lambda(\lambda+2)}{\sigma^{2}}+\frac{1}{(\sigma-1)^{2}}}\right]}{\gamma(b-1)(\lambda+2)}.
\]
Given that the leading coefficient of the polynomial is negative,
we have that $\Omega(\phi)$ is increasing for $\phi\in\left(0,\phi^{+}\right)$
and decreasing for $\phi\in\left(\phi^{+},1\right)$. Thus, $\Omega(\phi)$
has at most two zeros for $\phi\in\left(0,1\right)$, called sustain
points $\phi_{s1}$ and $\phi_{s2}$ (with $\phi_{s1}<\phi_{s2}$).\footnote{One of which is given by $\phi=1$ if $b=\frac{1}{2}$.}
If $b<\tfrac{1}{2}$, there exist at most two sustain points $\phi_{s1}\in\left(0,1\right)$
and $\phi_{s2}\in\left(0,1\right)$ and we have $\Omega(\phi)<0$
for $\phi\in\left\{ \left(0,\phi_{1s}\right)\cup\left(\phi_{2s},1\right)\right\} $
and $\Omega(\phi)>0$ for $\phi\in(\phi_{1s},\phi_{2s})$. If $b>\tfrac{1}{2}$,
there exists one unique sustain point $\phi_{s1}\in\left(0,1\right)$
and we have $\Omega(\phi)<0$ for $\phi\in\left(0,\phi_{1s}\right)$
and $\Omega(\phi)>0$ for $\phi\in\left(\phi_{1s},1\right)$, which
concludes the proof.\hfill{}$\square$

\subsection*{B.4 Symmetric dispersion}

Using (\ref{eq:stability symmetric dispersion}), the break points
are given by: {\small{} 
\begin{align}
\phi_{b1} & =\frac{\sqrt{\gamma^{2}(\sigma-1)^{2}\left[8b(\lambda+1)^{2}-4\lambda^{2}-8\lambda-3\right]+4\gamma\sigma(\sigma-1)+4\sigma^{2}}-2b\gamma(\lambda+1)(\sigma-1)}{\gamma(\sigma-1)\left[2b(\lambda+1)-2\lambda-3\right]-2\sigma}\nonumber \\
\phi_{b2} & =-\frac{\sqrt{\gamma^{2}(\sigma-1)^{2}\left[8b(\lambda+1)^{2}-4\lambda^{2}-8\lambda-3\right]+4\gamma\sigma(\sigma-1)+4\sigma^{2}}+2b\gamma(\lambda+1)(\sigma-1)}{\gamma(\sigma-1)\left[2b(\lambda+1)-2\lambda-3\right]-2\sigma}.\label{eq:break points}
\end{align}
}{\small\par}

\noindent Define $\gamma_{1}$ as 
\[
\gamma_{1}\equiv\dfrac{2\sigma}{(2\lambda+1)(\sigma-1)},
\]
and $b_{1}$ and $b_{2}$ as 
\begin{align*}
b_{1} & \equiv\frac{\left[\gamma(2\lambda+1)(\sigma-1)-2\sigma\right]\left[\gamma(2\lambda+3)(\sigma-1)+2\sigma\right]}{8\gamma^{2}(\lambda+1)^{2}(\sigma-1)^{2}},\\
b_{2} & \equiv\frac{\gamma(2\lambda+1)(\sigma-1)-2\sigma}{2\gamma(\lambda+1)(\sigma-1)}.
\end{align*}
We have $b_{1}\in\left(0,\tfrac{1}{2}\right)$ and $b_{2}\in\left(b_{1},1\right)$.
The numerator of both break points in (\ref{eq:break points}) is
negative. The numerator of $\phi_{b1}$ can be shown to be positive
if $\gamma\leq\gamma_{1}$, implying $\phi_{b1}<0$. If $\gamma>\gamma_{1},$
then $\phi_{b1}>0$ if and only if $b\in[b_{1},b_{2})$, and $\phi_{b1}<1$
if and only if $b>b_{1}$. Therefore we have $\phi_{b1}\in(0,1)$
if and only if: 
\begin{enumerate}
\item $\gamma>\gamma_{1}$, 
\item $b\in\left[b_{1},b_{2}\right).$ 
\end{enumerate}
\noindent If, additionally, $b<\tfrac{1}{2}$, then we have also that
$\phi_{b2}\in\left(0,1\right)$.\footnote{If $\gamma\in\left(\gamma_{1},\frac{2\sigma}{\lambda(\sigma-1)}\right)$,
then $b_{2}\in\left(b_{1},\frac{1}{2}\right)$ and the condition is
trivially met by (ii). In this case, both break points exist.} If $\gamma\leq\gamma_{1},$ then only $\phi_{b2}$ exists when $b<\tfrac{1}{2}$.
If $b>\frac{1}{2}$, then $\phi_{b2}$ does not exist and we have
a single break point $\phi_{b1}$ that lies on the interval $(0,1)$
if conditions (i) and (ii) are satisfied. This means that the possibility
of complete re-dispersion following agglomeration as $\phi$ increases
requires that the relative weight of local spillovers is neither too
high below $0.5$ nor too low above zero.

\subsection*{B.5 Proof of Proposition 5}

Taking the derivative of $\mathcal{G}$ in (\ref{eq:stabasdisp})
with respect to $b$ we get: 
\[
\frac{\partial\mathcal{G}}{\partial b}=-2\gamma(\sigma-1)(2z-1)^{3}\left(1-\phi^{2}\right)<0,\ \ z\in\left(\frac{1}{2},1\right)
\]
Next, solving $\mathcal{G}$ in (\ref{eq:stabasdisp}) for $b$ yields:{\footnotesize{}
\begin{equation}
b=b_{c}\equiv\frac{(2z-1)\left(1-\phi^{2}\right)\left[\sigma+\gamma(\sigma-1)(1-2z)^{2}\right]-\sigma\left[2z^{2}(\phi-1)^{2}-2z(\phi-1)^{2}+\phi^{2}+1\right]\ln\left[\frac{z(\phi-1)+1}{z(1-\phi)+\phi}\right]}{2\gamma(\sigma-1)(2z-1)^{3}\left(1-\phi^{2}\right)}.\label{eq:bcritconda}
\end{equation}
}{\footnotesize\par}

\noindent As a result, we have $\mathcal{G}>0$ for $b<b_{c}$ and
$\mathcal{G}<0$ for $b>b_{c}$. Next, we will prove that $b_{c}<1/2$.

First, notice that $\lim_{\phi\rightarrow1}b_{c}=\frac{1}{2}.$ Next,
we have: 
\[
\dfrac{\partial b_{c}}{\partial\phi}=\frac{\sigma\mathcal{N}}{2\gamma(\sigma-1)(2z-1)^{3}\left(\phi^{2}-1\right)^{2}\left[z(\phi-1)+1\right]\left[z(\phi-1)-\phi\right]},
\]
where: 
\begin{align*}
\mathcal{N}= & (2z-1)\left(\phi^{2}-1\right)\left[2z^{2}(\phi-1)^{2}-2z(\phi-1)^{2}+\phi^{2}+1\right]-\\
 & -4\left[z(\phi-1)+1\right]{}^{2}\left[z(1-\phi)+\phi\right]{}^{2}\ln\left[\frac{z(\phi-1)+1}{z(1-\phi)+\phi}\right].
\end{align*}
The numerator of the derivative is negative. As for $\mathcal{N}$,
observe that: 
\begin{align*}
\dfrac{\partial\mathcal{N}}{\partial z}= & -4\sigma(2z-1)(\phi-1)^{2}\times\\
 & \times\left\{ (2z-1)\left(1-\phi^{2}\right)+2\left[z(\phi-1)+1\right]\left[z(\phi-1)-\phi\right]\ln\left[\frac{z(\phi-1)+1}{z(1-\phi)+\phi}\right]\right\} .
\end{align*}
The first term inside the curly brackets is positive and the second
one is negative as is the log term. Therefore, we have $\frac{\partial\mathcal{N}}{\partial z}<0$.
Since $\mathcal{N}\left(z=\frac{1}{2}\right)=0$ and given that $\mathcal{N}$
is continuous in $z$, we can conclude that $\mathcal{N}<0$ for $z\in\left(\frac{1}{2},1\right)$.
Thus, we have $\frac{\partial b_{c}}{\partial\phi}>0,$ which means
that $b_{c}<\frac{1}{2}$. Thus, if $b>\frac{1}{2}$, we have $\mathcal{G}<0$
for any value of $\lambda$ such that $\lambda^{*}(z)>0$. This concludes
the proof.

\subsection*{B.6 Bifurcation at symmetric dispersion}

\noindent We can get a better picture of the dynamic properties of
the model by studying the type of local bifurcation that the symmetric
equilibrium undergoes at some break-point $\phi=\phi_{b}.$ After
some tedious calculations, it is possible to show the following: 
\[
\dfrac{\partial f}{\partial z}\left(\dfrac{1}{2};\phi_{b}\right)=0;\ \dfrac{\partial^{2}f}{\partial z^{2}}\left(\dfrac{1}{2};\phi_{b}\right)=0;\ \dfrac{\partial f}{\partial\phi}\left(\dfrac{1}{2};\phi_{b}\right)=0;\ \dfrac{\partial^{2}f}{\partial\phi\partial z}\left(\dfrac{1}{2};\phi_{b}\right)>0,
\]
where $\phi_{b}\in\{\phi_{b1},\phi_{b2}\}$. According to \citeauthor{Guckenheimer2002}
(2002, pp.~150), the conditions above ensure that symmetric dispersion
undergoes a pitchfork bifurcation at $\phi=\phi_{b}$. Further, we
have{\small{} 
\[
\dfrac{\partial^{3}f}{\partial z^{3}}\left(\dfrac{1}{2};\phi_{b}\right)=\frac{32\mu(1-\phi_{b})}{(\sigma-1)\sigma(\phi_{b}+1)^{4}}\varsigma,
\]
}where 
\begin{equation}
\varsigma=\sigma(\phi_{b}-1)^{2}(\phi_{b}+1)-3\gamma(\sigma-1)\left[b(\phi_{b}+1)^{2}-\phi_{b}^{2}-1\right]\left[\lambda(\phi_{b}-1)+2\phi_{b}\right].\label{eq:criticality}
\end{equation}
If $\varsigma>0$, the derivative is positive, and the pitchfork is
subcritical and a curve of unstable asymmetric equilibria branches
from symmetric dispersion to its right. If $\varsigma<0$, the derivative
is negative and the pitchfork is supercritical and a curve of stable
asymmetric equilibria branches from symmetric dispersion. If $\varsigma=0$
we say that the pitchfork is degenerate. 

Equating $\varsigma$ to zero and solving for $b$ we have that $\varsigma<0$
if and only if 
\begin{equation}
b<b^{*}\equiv\frac{\frac{\sigma(\phi_{b}+1)(\phi_{b}-1)^{2}}{\gamma(\sigma-1)\left[\lambda(\phi_{b}-1)+2\phi_{b}\right]}+3\left(\phi_{b}^{2}+1\right)}{3(\phi_{b}+1)^{2}}.\label{eq:last-1}
\end{equation}
If the numerator of $b^{*}$ in (\ref{eq:last-1}) is negative, then
$b^{*}<0$ and $\varsigma>0$ for all $b\in(0,1)$. This happens if
$\gamma<\tilde{\gamma}$, where 
\[
\tilde{\gamma}\equiv-\frac{\sigma(\phi_{b}-1)^{2}(\phi_{b}+1)}{3(\sigma-1)\left(\phi_{b}^{2}+1\right)\left[\lambda(\phi_{b}-1)+2\phi_{b}\right]}.
\]
Since we are interested in the branching of the symmetric equilibrium
into a curve of asymmetric dispersion equilibria in $(\phi,z)\in(0,1)\times[0,1]$,
we assume that $b<\frac{\phi_{b}^{2}+1}{(\phi_{b}+1)^{2}}$, in accordance
with Proposition 2. We have $b^{*}<\frac{\phi_{b}^{2}+1}{(\phi_{b}+1)^{2}}$
if and only if 
\begin{align*}
\lambda & >\frac{2\phi_{b}}{1-\phi_{b}}.
\end{align*}
Note that $\tilde{\gamma}>0$ when $\lambda>\frac{2\phi_{b}}{1-\phi_{b}}$.
Thus, we have that $\varsigma<0$ if and only if (i). $\lambda>\frac{2\phi_{b}}{1-\phi_{b}}$,
(ii). $\gamma>\tilde{\gamma},$ and (iii). $b<b^{*}$. In these cases,
the pitchfork is supercritical. Otherwise, we have $\varsigma>0$
and the pitchfork is subcritical.

\bigskip{}

\noindent\textbf{\emph{Criticality of bifurcation when $b>\tfrac{1}{2}$.}}{} 

\noindent Let $b_{b}$ correspond to the level of $b$ such that $\dfrac{\partial f}{\partial z}\left(\dfrac{1}{2};b_{b}\right)=0$.
This implies equating $\mathcal{B}$ in (\ref{eq:stability symmetric dispersion})
to zero, yielding 
\[
b_{b}\equiv\frac{\gamma(\sigma-1)\left[(2\lambda+3)\phi^{2}+2\lambda+1\right]+2\sigma\left(\phi^{2}-1\right)}{2\gamma(\lambda+1)(\sigma-1)(\phi+1)^{2}}.
\]
Substituting $b_{b}$ in $\varsigma$ using (\ref{eq:last-1}) yields
\[
\left.\varsigma\right|_{b=b_{b}}=\frac{16\mu(\phi-1)^{2}\left\{ 3\gamma(\sigma-1)\left[\lambda(\phi-1)+2\phi\right]+2\sigma\left[2\lambda(\phi-1)+5\phi+1\right]\right\} }{(\lambda+1)(\sigma-1)\sigma(\phi+1)^{3}},
\]
which can be shown to be negative if $b_{b}>\tfrac{1}{2}.$ This means
that the pitchfork bifurcation is supercritical when $b\in\left(\tfrac{1}{2},1\right)$
and the break point $\phi_{b1}$ exists. Thus, a curve of stable asymmetric
equilibria branches from the symmetric dispersion at $\phi=\phi_{b1}$
and lies to its right.

\clearpage\singlespacing  \bibliographystyle{apalike}
\bibliography{refs}

\end{document}